\documentclass[sip,biber]{now-journal} %

\usepackage{amsmath,amssymb,amsfonts}
\usepackage{graphicx,color}
\usepackage{textcomp}
\usepackage{xcolor}
\usepackage{algorithm2e}

\usepackage{array}
\usepackage[caption=false,font=footnotesize,labelfont=sf,textfont=sf]{subfig}
\usepackage{bm}
\usepackage{bbm}
\usepackage{amsthm}
\usepackage{algpseudocode}
\usepackage{multirow}

\RestyleAlgo{ruled}
\graphicspath{{figs/}}

\newcommand{\dv}{\mathbf{d}}

\newcommand{\pv}{\mathbf{p}}

\newcommand{\uv}{\mathbf{u}}
\newcommand{\wv}{\mathbf{w}}

\newcommand{\xv}{\mathbf{x}}
\newcommand{\yv}{\mathbf{y}}

\newcommand{\Am}{\mathbf{A}}
\newcommand{\Bm}{\mathbf{B}}

\newcommand{\Dm}{\mathbf{D}}

\newcommand{\Hm}{\mathbf{H}}
\newcommand{\Id}{\mathbf{I}}

\newcommand{\Km}{\mathbf{K}}
\newcommand{\Lm}{\mathbf{L}}

\newcommand{\Qm}{\mathbf{Q}}

\newcommand{\Sm}{\mathbf{S}}

\newcommand{\Um}{\mathbf{U}}
\newcommand{\Wm}{\mathbf{W}}
\newcommand{\Vm}{\mathbf{V}}
\newcommand{\Xm}{\mathbf{X}}

\newcommand{\Ec}{\mathcal{E}}

\newcommand{\Gc}{\mathcal{G}}

\newcommand{\Nc}{\mathcal{N}}
\newcommand{\Oc}{\mathcal{O}}

\newcommand{\Vc}{\mathcal{V}}

\newcommand{\piv}{\hbox{\boldmath$\pi$}}

\newcommand{\Lambdam}{\hbox{\boldmath$\Lambda$}}

\title{Lossy compression of weighted graph adjacency matrices by transform coding}

\author[1*]{Yanagiya,Kenta}
\author[1]{Hara,Junya}
\author[2]{Higashi,Hiroshi}
\author[1]{Tanaka,Yuichi}
\author[3]{Ortega,Antonio}

\affil[1]{Graduate School of Engineering, The University of Osaka, Osaka, Japan}
\affil[2]{Faculty of Engineering Science, Kansai University, Osaka, Japan}
\affil[3]{Department of Electrical and Computer Engineering, University of Southern California, CA, USA}

\issuevolumeyear{2026}
\issuevolumenumber{xx}
\issueelocation{e} %
\articletype{Original Paper} %

\articledatabox{Received xx xxxxx 2024; Revised xx xxxxx 2024\\[2pt]
ISSN 2048-7703; DOI 10.1561/116.xxxxxxxx\\
\copyright\ 2026 K. Yanagiya, J. Hara, H. Higashi, Y. Tanaka and A. Ortega}

\OAstatement{This is an Open Access article, distributed under the terms of the Creative Commons Attribution licence (\url{http://creativecommons.org/licenses/by-nc/4.0/}), which permits unrestricted re-use, distribution, and reproduction in any medium, for non-commercial use, provided the original work is properly cited.}

\keywords{Graph signal processing, graph filter bank, line graph, edge smoothness}

\creditline*{Corresponding author: Kenta Yanagiya, k\_yanagiya@msp-lab.org. 
The preliminary version of this paper is presented in \citep{yanagiya2024a}.
This work was supported in part by JSPS KAKENHI under Grant 23H01415, JST SPRING under Grant JPMJSP2138, and JST AdCORP under Grant JPMJKB2307.}

\begin{document}

\begin{abstract}
In this paper, we propose a compression framework for weighted graphs in which the graph topology is transmitted losslessly and edge weights are compressed lossily.
A challenge in the lossy compression of edge weights is that the underlying relationships between edges are ambiguous.
To address this issue, we first transform the unweighted graph into the corresponding \textit{line graph}, whose nodes represent the edges of the original graph and whose edges encode the relationships between them.
The line graph transform allows us to regard edge weights as a graph signal defined on the line graph.
Instead of transmitting the edge-weight vector, we first transform it with a graph filter bank on the line graph.
Then, quantization and entropy coding are performed on the transformed coefficients of the edge weight vector.
In addition to the lossy compression method, we formalize edge smoothness on the line graph and show that it serves as a measure of the difficulty of compression. 
The proposed smoothness measure can be easily calculated without converting to a line graph.
This provides insight into the expected compression performance of a given weighted graph.
Experiments on synthetic and real-world data validate the effectiveness of the proposed method by comparing it with existing matrix preprocessing methods.
\end{abstract}

\section{Introduction}\label{sec: intro}
Graphs are powerful mathematical representations of the connectivity among nodes and are widely used to analyze various networks, including social, brain, sensor, and transportation networks \citep{ortega2018, shuman2013, tanaka2018, tanaka2020, yamada2020, sakiyama2019a, onuki2016}.
While graphs are useful, they often have many nodes and edges. 
This demands significant computational and storage resources for storing and/or transmitting them \citep{besta2019}.
Therefore, effective compression methods for weighted graphs are crucial.

Several lossless graph compression schemes have been proposed for unweighted graphs, including conversions to bitmap, $k^2$ trees, and succinct data structures \citep{besta2019, alvarez2010, khandelwal2017, martinez-bazan2012}. These methods are generally efficient (although they may be reaching their theoretical performance limits \citep{besta2018}) and desirable (since topological information can be fully preserved). 
However, efficient lossless compression of weighted graphs is challenging because weights are typically real-valued, although it can be achieved if the graphs possess special properties, such as sparsity \citep{besta2019, toivonen2011}.
To address this challenge, this paper proposes focusing on \textit{lossy graph compression for weighted graphs}, in which the topology (corresponding to the unweighted graph) is encoded losslessly.  

Many lossy graph compression methods have been proposed, including node and edge sampling, graph sparsification, and graph coarsening \citep{sakiyama2019a, yanagiya2022, spielman2011a, jin2020}, which are typically based on \textit{graph simplification}. 
In weighted graphs, edges convey information about (1) the existence of interconnections between nodes (edge on/off) and (2)  the strength of edges (edge weights). 
Most lossy compression approaches for weighted graphs consider both edge existence and edge weights simultaneously during simplification \citep{willcock2006, Zhou2012, toivonen2011, liu2012}, which can reduce the number of edges and, in some cases, nodes. 
The performance of these methods is often quantified in terms of errors in edge weights and degree distribution.
However, these methods modify the original graph topology, which is generally undesirable because it can significantly affect graph-based applications and graph visualization.
Additionally, these changes alter the graph spectrum and can therefore significantly impact various graph signal processing (GSP) methods \citep{cheung2018, hara2023, hara2023a}.
As an alternative, we propose that compression should consider not only edge-weight errors but also their impact on downstream tasks that use the graph spectrum, such as graph signal diffusion and spectral clustering \citep{tanaka2020, onuki2016, vonluxburg2007}.

Our work is motivated by the assumption that topological information plays a more significant role in many applications than edge weights, leading to the idea of compressing them separately \citep{caldarelli2012}.
This is different from many existing graph compression or simplification methods, which operate directly on the topology \citep{besta2019, jin2020}.
 
Note that many works have proposed \textit{graph signal} transforms \citep{narang2013, sakiyama2014, sakiyama2019, pavez2023}, which are designed to sparsely represent signals on a graph.
However, they are not intended for the \textit{graph} itself (i.e., its adjacency matrix).
Therefore, there is no canonical way to apply existing graph signal transforms directly to adjacency matrices so as to obtain sparse coefficient representations.
Edge weights can be represented as a vector, but there is no unique way to do so. 
Different node labelings induce different permutations of the edge-weight vector, and standard transforms generally produce different coefficients for different permutations. 
Therefore, one challenge in compressing edge weights is organizing them to preserve inter-edge relationships.

In this paper, we address the challenges outlined above in graph compression.
First, to handle permutation invariance, we propose a two-step compression method: the binary topology information is compressed losslessly, while the edge weights are compressed lossily.
Since both the encoder and decoder have access to the losslessly encoded graph topology, edge-weight compression can be based on it.
Specifically, we convert the original adjacency information into a \textit{line graph}, whose nodes and edges represent the edges of the original graph and their relationships, respectively \citep{aigner1967, evans2010}.
With this conversion, \textit{edge weights are regarded as a graph signal on the line graph}.
We can then use various graph signal transforms that explicitly account for edge positions to obtain a sparse representation of the edge weights.

Second, we demonstrate that many real-world edge weights are smooth on their corresponding line graphs, indicating that the proposed compression framework is well-suited to practical applications.
Furthermore, we show that this edge smoothness can be computed efficiently without explicitly constructing the line graph: it only requires node-wise statistics on the original graph.
This efficient computation allows us to use edge smoothness as a measure of compression difficulty, providing insight into the expected compression performance for a given weighted graph.

We perform graph compression experiments using synthetic and real-world weighted graphs.
For synthetic graphs, we demonstrate that the proposed compression method achieves better reconstruction accuracy than alternative methods across various levels of edge smoothness.
We also demonstrate that the proposed method outperforms alternatives on real-world graphs.

In this paper, we provide several substantial additions beyond our conference paper \citep{yanagiya2024a}:
\begin{enumerate}
    \item We formalize \emph{edge smoothness} on the line graph and propose a local--global decomposition that enables computing the variation in $\Oc(N^2)$ without explicitly forming the line graph (see Section~\ref{sec: edge_smoothness_sub} and Section~\ref{sec: edge_smoothness_original_graph});
    \item We demonstrate that edge smoothness serves as an effective measure of \emph{compression difficulty}, providing a theoretical basis for predicting compression performance (see Section~\ref{sec: compressibilty});
    \item We expand the experiments to three synthetic graph families, $6$ real-world traffic networks, and $3$ power-grid with non-geometric edge weights, including ablations over quantization steps (see Section~\ref{sec: experiments}).
\end{enumerate}

The remainder of this paper is organized as follows.
Section \ref{sec: preliminary} introduces the notation, the line graph conversion, and graph operators for inter-edge relations.
Section \ref{sec: lossy_compression} presents the proposed compression framework and the transform-coding pipeline for edge weights.
Section \ref{sec: edge_smoothness} formalizes the edge-domain smoothness assumption, provides its interpretations on both the original graph and the line graph, and demonstrates that edge smoothness serves as a measure of compression difficulty. Section \ref{sec: experiments} presents numerical experiments on synthetic and real data.
Finally, Section \ref{sec: conclusion} concludes the paper.

\section{Preliminary}\label{sec: preliminary}
\subsection{Notation}
In this paper, we consider that the weighted undirected graph is represented as $\Gc = (\Vc, \Ec, \Wm)$, where $\Vc$ and $\Ec$ are the sets of nodes and edges, respectively, and $\Wm$ is the weighted adjacency matrix, where the $(i,j)$-entry $w_{ij}$ represents the edge weights between node $i$ and node $j$.
$N = |\Vc|$ is the number of nodes.
All-one vectors with $N$ and $|\Ec|$ elements are denoted by $\mathbf{1}_N$ and $\mathbf{1}_{|\Ec|}$, respectively, the identity matrices with $N$ and $|\Ec|$ diagonal elements are $\Id_N$ and $\Id_{|\Ec|}$.
We define two types of degree matrices: the unweighted degree matrix $[\Dm]_i = [\Am\mathbf{1}_N]_i$, which counts the number of edges incident to each node, and the weighted degree matrix $[\Km]_i = [\Wm\mathbf{1}_N]_i$, which sums the edge weights incident to each node.
The combinatorial graph Laplacian matrix $\Lm = \Km - \Wm$, a real symmetric matrix, has non-negative real eigenvalues $\lambda$ and corresponding orthonormal eigenvectors $\uv_i$ so that we can write $\Lm = \Um\bm{\Lambda}\Um^\top$, where $\Um$ is a orthonormal matrix and $\bm{\Lambda} = \text{diag}(\lambda_0, \cdots, \lambda_{N-1})$.
Without loss of generality, these eigenvalues are assumed to be ordered as $0=\lambda_{0}<\lambda_{1} \leq \lambda_{2} \leq \cdots \leq \lambda_{N-1}=\lambda_{\max}$.
For line graphs, we use the undirected incidence matrix:
\begin{equation}\label{eqn: undirected_incidence_matrix}
    [\tilde{\Bm}]_{i\alpha} = \begin{cases}
    1 & \text{Edge } \alpha=(i,j) \textrm{ is incident to node }\; i \; \textrm{or} \; j,\\
    0 & \text{Otherwise}.
\end{cases}
\end{equation}

\subsection{Graph Signal Smoothness}
A graph signal $x: \Vc \rightarrow \mathbb{R}$ is a function that assigns a real value to each node.
Graph signals can be written as vectors $\xv \in \mathbb{R}^N$ whose $i$th element, $[\xv]_i$, represents the signal value at the $i$th node. The graph Fourier transform (GFT) of a graph signal is defined by $\hat{\xv} = \Um^\top \xv$.
In GSP, \textit{smoothness of graph signals} refers to a small variation of signal values between adjacent nodes over the entire graph \citep{tanaka2020, chen2015}.
The graph signal variation $\Delta_{\Lm}(\xv)$ of a graph signal $\xv$ on $\Gc$ is calculated as follows:
\begin{equation}\label{eqn: signal_smoothness}
    \Delta_{\Lm}(\xv) = \xv^\top\Lm\xv = \sum_{(i,j)\in \Ec} [\Wm]_{ij}([\xv]_i - [\xv]_j)^2.
\end{equation}
Smooth graph signals have energy concentrated in the coefficients $\hat{\xv}$ corresponding to the low-frequency band and are known to be sparse in the graph frequency domain.  
Therefore, the graph signal variation \eqref{eqn: signal_smoothness} can be viewed as a measurement of compressibility for graph signals.

\begin{figure}[tp]
    \centering
    \includegraphics[width=\linewidth]{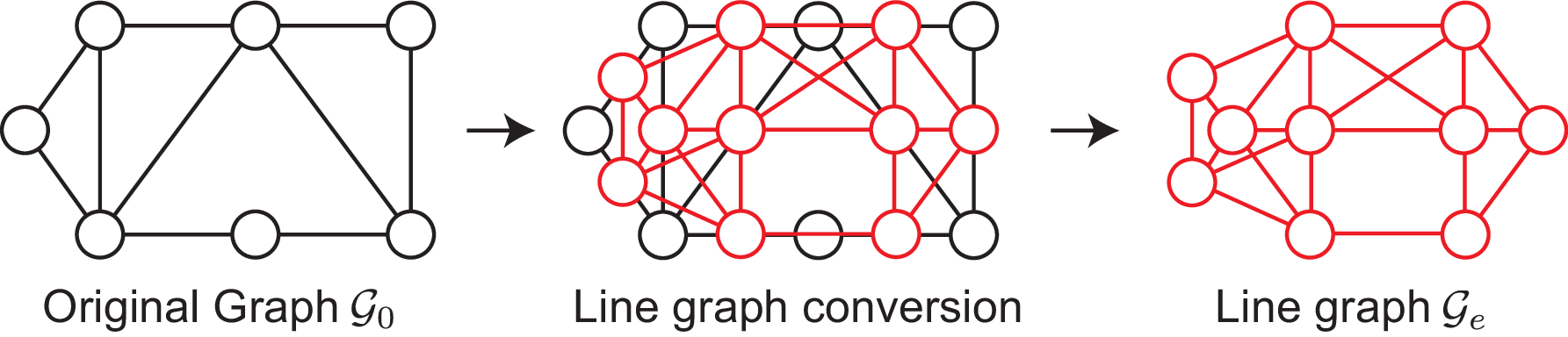}
    \caption{Illustration of line graph conversion. The red circles correspond to edges of the original graph. An edge of a line graph is drawn if the two edges are connected via a node in the original graph.}
    \label{fig: line_graph_conversion}
\end{figure}

\subsection{Line Graph Conversion and Graph Operator for Inter-edge Relationships}\label{sec: line graph}
In this paper, we only consider unweighted line graphs. 
\begin{definition}[Line graph]
    \label{def: Line Graph}
    For a given graph $\Gc = (\Vc, \Ec, \Wm)$, the unweighted adjacency matrix of the corresponding line graph $\Gc_e =(\Vc_e, \Ec_e, \Am_e)$ is defined as follows \citep{evans2010}:
    \begin{equation}
    \Am_e = \tilde{\Bm}^\top\tilde{\Bm} - 2\Id_{|\Ec|}.
    \label{eqn: adj_line}
    \end{equation}
\end{definition}
\noindent
The line graph conversion process is illustrated in Fig.~\ref{fig: line_graph_conversion}. Note that in the line graph $\Gc_e$ we have $|\Vc_e| = |\Ec|$.
Since $\Gc_e$ is unweighted by definition, its corresponding graph Laplacian is 
\begin{equation}\label{eqn: lap_line}
    \Lm_e = \Dm_e - \Am_e,
\end{equation}
where $\Dm_e$ is the unweighted degree matrix of $\Gc_e$.
The graph Laplacian of a line graph explicitly captures the interaction between the nodes of the line graph, that is, between the edges of the original graph.
Similar to $\Lm$, we can define the eigenvalue decomposition of $\Lm_e$ as follows:
\begin{equation}\label{eqn: eigenvalue_decomposition_line}
    \Lm_e = \Vm \Lambdam_e \Vm^\top,
\end{equation}
where $\Lambdam_e$ and $\Vm$ are eigenvalue matrix of $\Lm_e$ and its corresponding eigenvector matrix, respectively.

\section{Lossy Compression Framework of Weighted Adjacency Matrices}\label{sec: lossy_compression}
In this section, we propose a lossy compression framework for weighted graphs.
First, we provide an overview of the framework, and then we discuss the computational cost of its building blocks.

\subsection{Framework}
The overview of the proposed weighted adjacency matrix compression is shown in Fig.~\ref{fig: Overview_Compression}.
We transmit the graph's binary adjacency information losslessly and compress the edge weights using a graph filter bank.
The proposed method has two main advantages: 1) the topological information is completely preserved, and 2) edge weight compression can be performed with bitrate scalability.

\begin{figure}[tp]
    \centering
    \includegraphics[width = \linewidth]{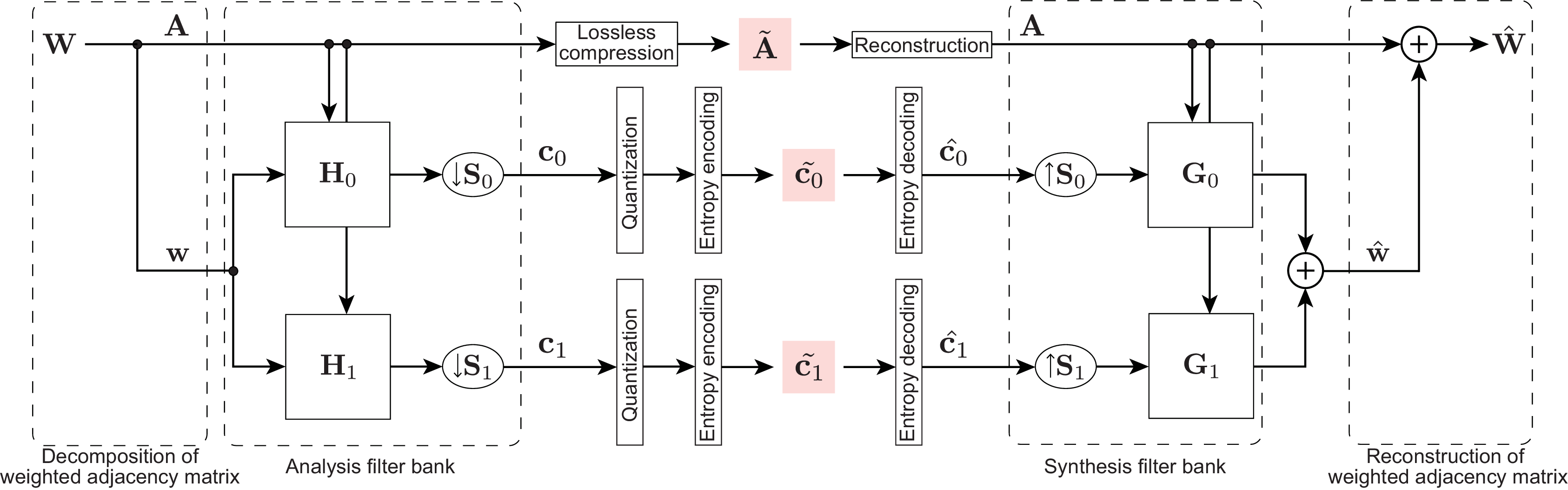}
    \caption{Overall framework of the proposed weighted adjacency matrix compression method. The compressed matrix and vector are highlighted in red.}
    \label{fig: Overview_Compression}
\end{figure}

More formally, a given weighted adjacency matrix $\Wm$ is decomposed into a binary adjacency matrix $\Am \in \{0,1\}^{N\times N}$ and an edge weight vector $\wv \in \mathbb{R}_{\ge 0}^{|\Ec|}$, where $\Am$ is given by:
\begin{equation}
    [\Am]_{ij} = \begin{cases}
        1 & [\Wm]_{ij} \neq 0,\\
        0 & \text{otherwise.}
    \end{cases}
\end{equation}

We can use one of the existing compression methods to compress $\Am$ losslessly \citep{besta2019, alvarez2010, khandelwal2017}.
Our main compression target is $\wv$, defined on the edges of the original graph, which we now treat as a graph signal on the line graph $\Gc_e$, which, as noted above, has $|\Vc_e| = |\Ec|$ nodes. 
In this paper, $\wv$ is transformed by a critically-sampled graph filter bank \citep{narang2013, sakiyama2019, pavez2023, narang2012} on $\Gc_e$.
After the analysis transform, the subband coefficients $\yv_c$ are obtained as follows:
\begin{equation}\label{eqn: subband_coeff}
    \yv_c = \Sm_c^\top\Hm_c\wv
\end{equation}
where $\Hm_c \in \mathbb{R}^{|\Ec|\times|\Ec|}$ is the $c$th analysis filter in the graph filter bank ($c = 0, 1, \ldots, M-1$) and $\Sm_c^\top \in \mathbb{R}^{|\Ec_c|\times|\Ec|}$ is the downsampling operator on the line graph.
Here, $|\Ec_c|$ denotes the number of retained coefficients in the $c$th subband.
The subband coefficients $\{\yv_c\}$ are then quantized and entropy coded.
Since each subband coefficient sequence $\{\yv_c\}_{0, 1, \ldots, M-1}$ is a one-dimensional vector in $\mathbb{R}^{|\Ec_c|}$, standard scalar quantization and entropy coding can be applied to each subband independently.
Note that our proposed framework is general: We can use any graph filter banks and sampling methods, including methods presented in \citep{narang2013, sakiyama2019, pavez2023, narang2012, sakiyama2016}.

At the decoder, the decoded edge weights are assigned to the edges specified by the losslessly decoded binary adjacency matrix.
Let $\hat{\wv}$ be the decoded edge-weight vector and let edge $\alpha$ correspond to $(i,j) \in \Ec$. 
The reconstructed weighted adjacency matrix is defined by
\begin{equation}\label{eqn: reconstruction_rule}
    [\hat{\Wm}]_{ij} = 
    \begin{cases}
        [\hat{\wv}]_{\alpha} & [\Am]_{ij} = 1,\\
        0 & [\Am]_{ij} = 0.
    \end{cases}
\end{equation}

In the proposed framework, the binary adjacency matrix is transmitted losslessly and remains available at the decoder. 
Note that if a particular edge is reconstructed with a weight of zero due to quantization, the corresponding edge is removed.

\subsection{Note on Computational Costs}
Since the number of edges in the original graph is generally larger than the number of nodes, i.e., $|\Ec| \geq N$, we must ensure that both filtering and sampling steps remain computationally tractable.
Some graph filter banks require an eigenvalue decomposition whose computational cost is $\Oc(|\Ec|^3)$ for the line graph \citep{shuman2013, pavez2023}.
With compact support filter banks, defined by a set of polynomial filters, the filtering step does not require eigendecomposition, reducing the computational cost to $\Oc(p|\Ec|^2)$, where $p$ is the order of the polynomial \citep{hammond2009}.

To reduce redundancy, the transform coefficients are subsampled as in \eqref{eqn: subband_coeff}.
Critical sampling in the node domain requires a computational complexity of $\Oc(|\Ec|)$.
Alternatively, sampling in the spectral domain \citep{sakiyama2019} also requires a $\Oc(|\Ec|)$, but requires eigendecomposition to convert the signal into the spectral domain, thus having overall $\Oc(|\Ec|^3)$ complexity.
In our experiments, we use several graph filter banks representative of various design techniques (Section \ref{sec: experiments}).
Two-channel graph filter banks with node-domain sampling that do not require eigenvalue decomposition are generally preferred for the proposed framework due to their lower complexity. 

\section{Edge Weight Smoothness and Compressibility of Edge Weights}\label{sec: edge_smoothness}
In this section, edge weight smoothness is studied within our framework (see Section \ref{sec: lossy_compression}).
We also demonstrate that many real-world weighted graphs have smooth edge weights and show that edge-weight smoothness can be used to quantify edge-weight compressibility.

\subsection{Edge Weight Smoothness}\label{sec: edge_smoothness_sub}
Suppose that the original graph $\Gc = (\Vc, \Ec, \Wm)$ and its corresponding line graph $\Gc_e = (\Vc_e, \Ec_e, \Am_e)$ are given, and consider the smoothness of edge weights for a given topology.
Since the edge weights can be viewed as a graph signal on the line graph $\Gc_e$, the edge weight variation can be naturally defined by applying the definition of \eqref{eqn: edge_variation} to the line graph.
\begin{definition}[Global Edge Weight Variation]\label{def: edge_weight_smoothness}
    The variation of the edge weight signal $\wv$ on the line graph $\Gc_e$ is given by
    \begin{equation}\label{eqn: edge_variation}
        \Delta_{\Lm_e}(\wv) = \wv^\top\Lm_e\wv = \sum_{(\alpha, \beta) \in \Ec_e} ([\wv]_\alpha - [\wv]_\beta)^2,
    \end{equation}
    where $\alpha$ and $\beta$ are edge indices of the original graph (and therefore node indices in $\Gc_e$), and $(\alpha, \beta) \in \Ec_e$ represents a pair of adjacent edges in the line graph. 
\end{definition}
\noindent
Note that if the original graph is unweighted, we have $\wv = \mathbf{1}_{|\Ec|}$, and therefore $\Delta_{\Lm_e}(\wv) = 0$, indicating that the edge signal is indeed smooth (since all edge weights are equal to one).
If $\Gc$ is weighted, $\wv$ will not be as smooth and $\Delta_{\Lm_e}(\wv)$ can be large. As we will see, this lack of smoothness means that more bits are required to compress $\wv$ while achieving a given degree of accuracy. 

\subsection{Experimental Validation of Edge Weight Smoothness}\label{sec: toy_example}
We now use the normalized edge weight variation:
\begin{equation}
\label{eq:normalized-variation}
    \underline{\Delta}_{\Lm_e}(\wv) = \Delta_{\Lm_e}(\wv)/\|\wv\|^2, 
\end{equation}
to evaluate the smoothness of edge weights in three networks:
\begin{enumerate}
    \item Synthetic random sensor graph ($N=100$ and $|\Ec| = 360.18$ on average), where a graph is constructed by placing $100$ nodes uniformly at random in space and connecting each node to its $6$ nearest neighbors.
    The edge weight $[\Wm]_{ij}$ is computed as 
    \begin{equation}\label{eqn: distance}
        [\Wm]_{ij} = \mathrm{exp}\left(-\frac{\|\pv_i-\pv_j\|_2}{0.3}\right),
    \end{equation}
    where $\pv_i \in \mathbb{R}^{2}$ is a 2-D coordinate of nodes.
    
    \item Traffic network ($N=322.27$ and $|\Ec| = 445.37$ on average): Graphs are obtained from the GSP-traffic dataset \citep{gsp_traffic} whose road networks are taken from OpenStreetMap.
    This dataset has 465 traffic networks in various regions.
    In this experiment, we use the 195 connected graphs in this dataset.
    The edge weights are calculated with \eqref{eqn: distance}.
    
    \item power-grid network ($N=294.33$ and $|\Ec| = 460.00$ on average): We also use a dataset comprising three graphs derived from power-grid data in Chile \citep{kim2018}.
    Each represents a distinct configuration of a power-grid.
    All components, like power plants, are modeled as nodes, and edge weights are determined by the voltage of the corresponding transmission line.
\end{enumerate}
Table \ref{tbl: ex-variation} provides $\underline{\Delta}_{\Lm_e}(\wv)$ for these three networks. 
As a baseline, for each physical graph topology, we generate random edge weights independently from a uniform distribution on [0, 1] and compare the normalized edge weight variation $\underline{\Delta}_{\Lm_e}(\wv)$ in the physical graph to that in the graph with random edge weights. 
For the random baseline, we generated $10$ independent realizations of random edge weights for each topology and report the average normalized variation over all realizations and graph instances.

Across all three datasets, the average normalized edge-weight variation of the physical graphs is lower than that of the random baseline. 
Although the magnitude of the reduction depends on the dataset, this consistent trend suggests that many realistic networks exhibit smooth variation in edge weights.

\begin{table}[tp]
    \caption{Some examples of normalized edge weight variation $\underline{\Delta}_{\Lm_e}(\wv)$ across the network. The table presents results for synthetic sensor networks with distance-based edge weights, traffic networks with distance-based edge weights, and power-grids with voltage as edge weights. The baseline is the variation of random edge weights generated from a uniform distribution.}
    \label{tbl: ex-variation}
    \centering
    \begin{tabular}{c|ccc}\hline
        Network Type & Sensor & Traffic \citep{gsp_traffic} & power-grid \citep{kim2018} \\ \hline\hline
        Number of graphs & 100 & 195 & 3 \\ \hline
        $\lambda_{e, \max}$ & 20.79 & 11.62 & 44.33 \\ \hline
        Random weights: avg. $\underline{\Delta}_{\Lm_e}(\wv)$ & 7.77 & 1.26 & 1.77 \\ \hline
        Physical weights: avg. $\underline{\Delta}_{\Lm_e}(\wv)$ & 5.99 & 0.68 & 1.47 \\ \hline
    \end{tabular}
\end{table}

\subsection{Local and Global Edge Weight Variations} \label{sec: edge_smoothness_original_graph}
The edge weight variation in Definition \ref{def: edge_weight_smoothness} can be directly derived from the original graph using the following proposition.
\begin{proposition}\label{prop: edge_smoothness_original_graph}
    Given a weighted graph $\Gc = (\Vc, \Ec, \Wm)$, $\Delta_{\Lm_e}(\wv)$ in \eqref{eqn: edge_variation} can be rewritten as 
    \begin{equation}\label{eqn: edge_smoothness_original_graph}
        \Delta_{\Lm_e}(\wv) = \sum_{i \in \Vc}[\dv]_i^2([\boldsymbol{\mu}_{w^2}]_i - [\boldsymbol{\mu}_w]_i^2) = \sum_{i \in \Vc} [\dv]_i^2[\boldsymbol{\sigma}^2_w]_i,
    \end{equation}
    where $[\dv]_i$ is the unweighted degree of node $i$ in $\Gc$, and we define $[\boldsymbol{\mu}_{\wv^2}]_i=\frac{1}{[\dv]_i}\sum_{j \in \Nc_i} [\Wm]^2_{ij}$, where $\Nc_i$ is the set of neighbors of node $i$, and $[\boldsymbol{\mu}_w]_i = \frac{1}{[\dv]_i}\sum_{j \in \Nc_i}[\Wm]_{ij}$.
\end{proposition}

\begin{proof}
    The edge weight variation \eqref{eqn: edge_variation} can be written as follows: 
    \begin{equation}
        \begin{aligned}
            \Delta_{\Lm_e}(\wv) &= \wv^\top \Dm_e \wv - \wv^\top\Am_e\wv\\
            &= \wv^\top \Dm_e \wv - \wv^\top\tilde{\Bm}^\top\tilde{\Bm}\wv + 2\wv^\top\wv\\
            &= \sum_{\alpha \in \Ec}([\Dm_e]_{\alpha} + 2)[\wv]_\alpha^2 -\sum_{i\in \Vc}\Bigl(\sum_{j \in \Nc_i} [\Wm]_{ij}\Bigr)^2\\
            &\overset{\text{(a)}}{=} \sum_{i\in \Vc}\biggl\{\sum_{j\in \Nc_i} \frac{[\dv]_i + [\dv]_j}{2}[\Wm]_{ij}^2 -\Bigl(\sum_{j \in \Nc_i} [\Wm]_{ij}\Bigr)^2\biggr\}\\
            &= \sum_{i \in \Vc}\biggl\{[\dv]_i \sum_{j \in \Nc_i} [\Wm]_{ij}^2 - \Bigl(\sum_{j \in \Nc_i} [\Wm]_{ij}\Bigr)^2\biggr\}\\
            &= \sum_{i \in \Vc}[\dv]_i^2([\boldsymbol{\mu}_{\wv^2}]_i - [\boldsymbol{\mu}_w]_i^2)\\
            &= \sum_{i \in \Vc}[\dv]_i^2 [\boldsymbol{\sigma}^2_w]_i,
        \end{aligned}
    \end{equation}
    where we use the relationship $[\Dm_e]_{\alpha} = [\dv]_i + [\dv]_j - 2$ for (a) \citep{wang2016a,yanagiya2025}.
\end{proof}

\noindent
Note that $[\dv]_i^2[\boldsymbol{\sigma}_w^2]_i$ in \eqref{eqn: edge_smoothness_original_graph} can be viewed as the \textit{local edge weight variation} at node $i$, which increases quadratically with the number of edges $[\dv]_i$ 
and with the variance of the local edge weights ($[\boldsymbol{\sigma}_w^2]_i$).
Interestingly, all the terms in \eqref{eqn: edge_smoothness_original_graph} can be easily calculated from local operations on the original adjacency matrix $\Wm$.
This avoids the computational cost of explicitly constructing the line graph. 

\subsection{Edge Weight Compressibility}\label{sec: compressibilty}
The decomposition in Proposition \ref{prop: edge_smoothness_original_graph} also provides a meaningful interpretation for compression.
Since the global edge weight variation $\Delta_{\Lm_e}(\wv)$ measures how much the edge weights vary across adjacent edges in the line graph, it directly relates to the compressibility of the edge weight vector $\wv$.

From transform coding theory, signals with low variation tend to have energy concentrated in low-frequency components, leading to sparse representations in the transform domain.
Conversely, signals with high variation require more bits to achieve the same reconstruction quality.
This principle applies directly to our framework: edge weights with small $\Delta_{\Lm_e}(\wv)$ are expected to be more compressible by the graph filter bank on the line graph.
Therefore, $\Delta_{\Lm_e}(\wv)$ can be used as a predictor of compression difficulty for a given weighted graph.

Edge weights with different levels of variation $\Delta_{\Lm_e}(\wv)$ are shown on the line graph corresponding to the original graph in Fig.~\ref{fig: toy_example}(a).
As $\Delta_{\Lm_e}(\wv)$ increases, adjacent edges exhibit larger differences in their weights. 
The quantitative relationship between edge-weight variation and compression performance is evaluated in Section \ref{sec: experiments}.

\begin{figure}[tp]
\centering
\subfloat[Original graph $\Gc$]{\includegraphics[width=0.24\linewidth]{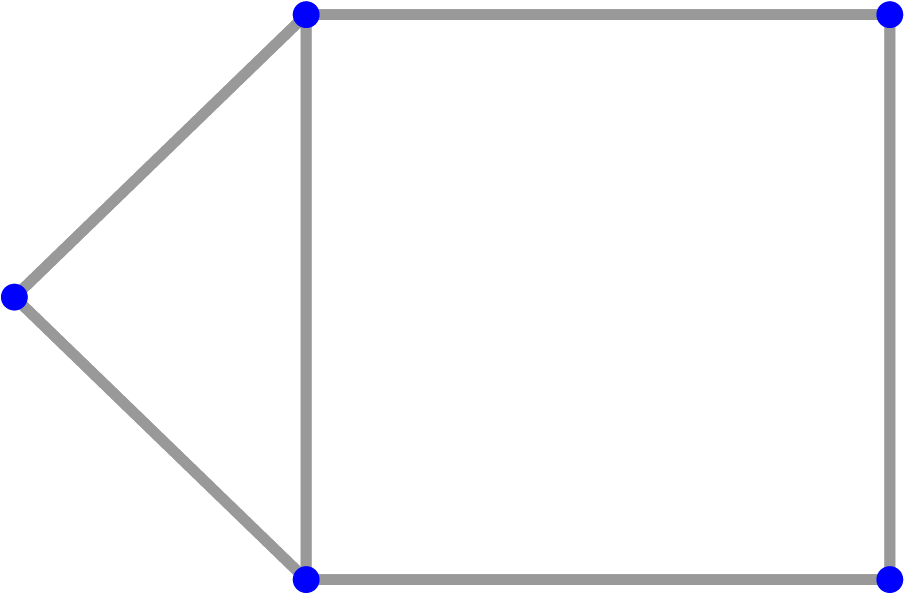}}
\subfloat[$\Delta_{\Lm_e}(\wv)=0.1$]{\includegraphics[width=0.24\linewidth]{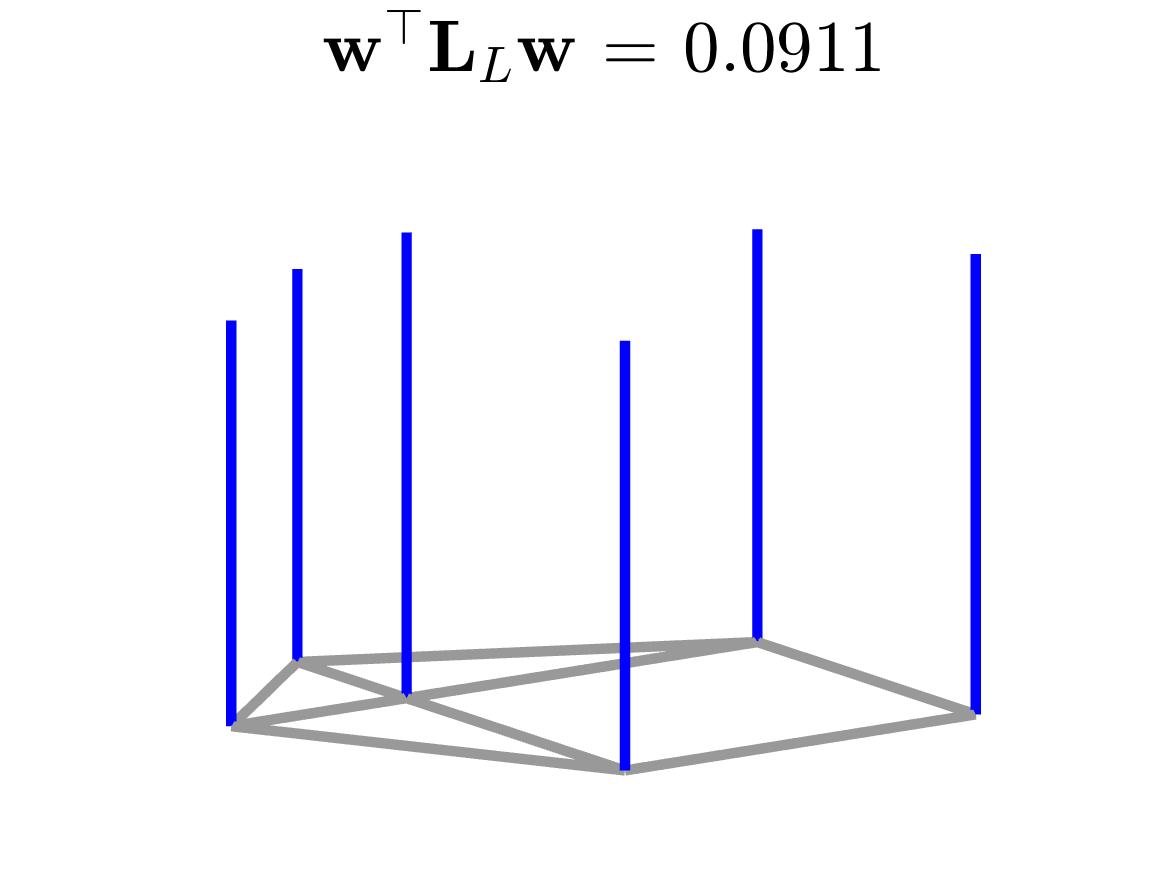}}
\subfloat[$\Delta_{\Lm_e}(\wv)=1.0$]{\includegraphics[width=0.24\linewidth]{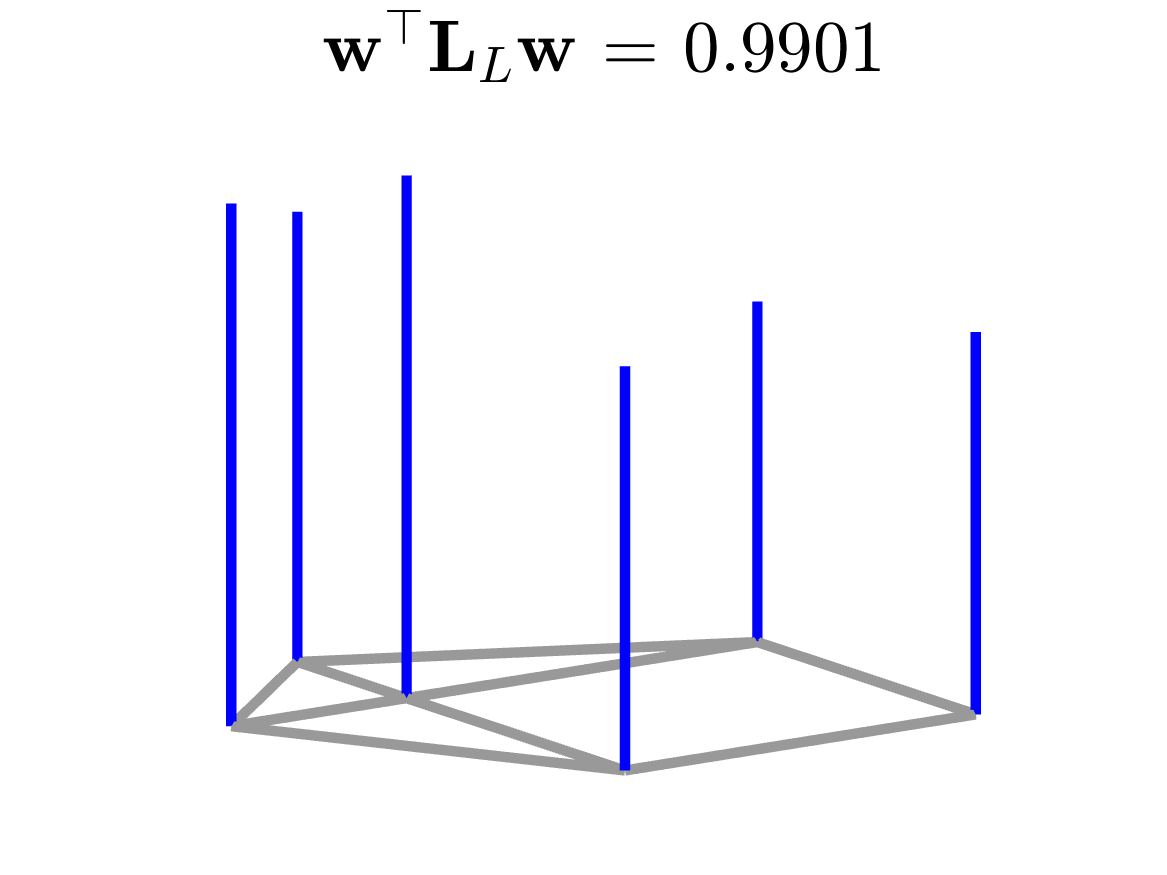}}
\subfloat[$\Delta_{\Lm_e}(\wv)=10$]{\includegraphics[width=0.24\linewidth]{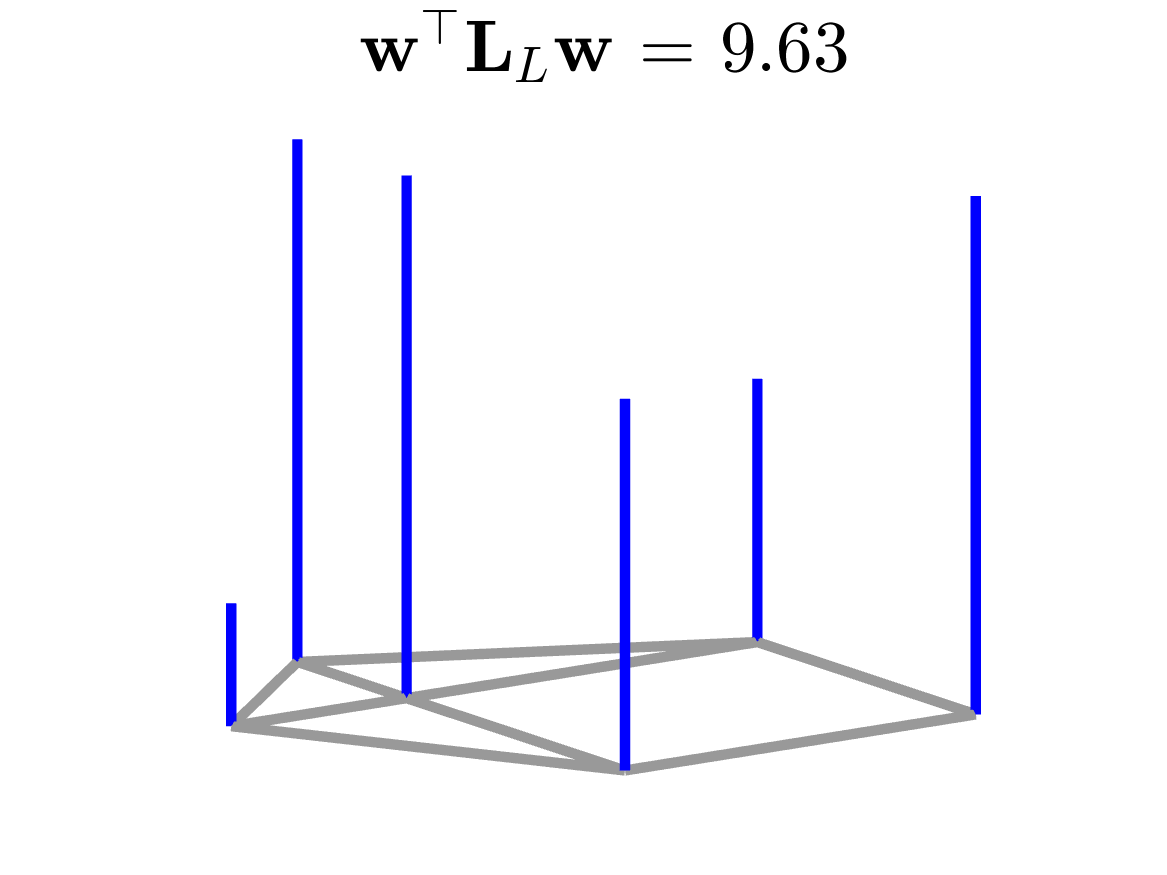}}
\caption{Illustration of edge weights with different smoothness levels on the same topology. (a) Original graph $\Gc$, (b)-(d) Edge weights visualized on the corresponding line graph $\Gc_e$ for three values of $\Delta_{\Lm_e}(\wv)$. The edge weight variation $\Delta_{\Lm_e}(\wv)$ increases from (b) to (d), indicating increasing compression difficulty. Smoother edge weights (b) exhibit nearly uniform values across adjacent edges, while less smooth edge weights (d) show significant variations.}
\label{fig: toy_example}
\end{figure}

\section{Experiments}\label{sec: experiments}
In this section, we conduct compression experiments on both synthetic and real-world weighted graphs and compare the proposed method with alternative ones.

\subsection{Synthetic Graphs}\label{sec: synthetic_experiments}
\subsubsection{Setup}
In the experiments on synthetic graphs, we use the following weighted graphs with $N  = 500$: random sensor graph with $|\Ec| = 1,781$, $10$-nearest-neighbor (10-NN) graph of \textit{Swiss Roll} point cloud with $|\Ec| = 4,188$, and random graph based on Erd\H{o}s--R\'{e}nyi model (ER model) with an edge probability $0.05$ and $|\Ec| = 3,097$.
Fig.~\ref{fig: graph} shows the synthetic graphs used.

\begin{figure}
    \centering
    \includegraphics[width=\linewidth]{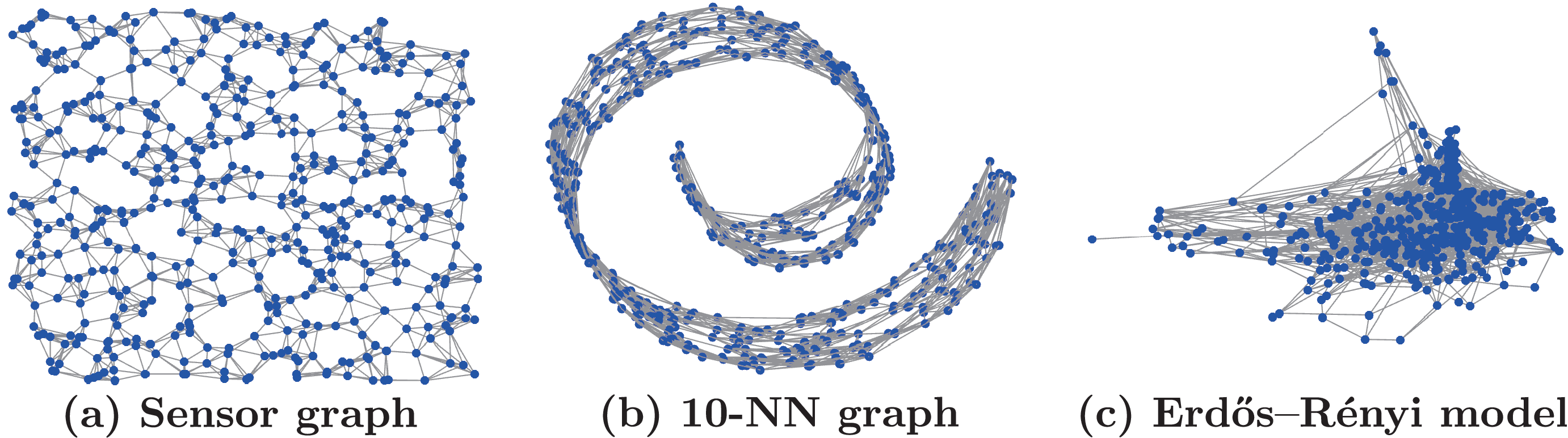}
    \caption{Examples of synthetic graphs with $N = 500$.}
    \label{fig: graph}
\end{figure}

The edge weights are generated according to different schemes depending on the graph type.
For the random sensor graph and the 10-NN graph, edge weights are determined by the distance-based function \eqref{eqn: distance}.
For the ER model, each edge weight is independently drawn from a truncated normal distribution $\Nc (\mu, \sigma^2)$ with $\mu = 1$ and variance $\sigma^2$ and the nonnegative interval $[0,\infty)$; the variance parameter is varied as $\sigma \in \{0.01, 0.05, 0.1, 0.5\}$ to investigate the relationship between edge smoothness and compression performance.

The proposed compression methods use the following critically sampled graph filter banks on the unweighted line graph: 
Graph-QMF filter bank with Meyer wavelets (abbreviated as \textbf{Proposed-QMF}) \citep{narang2012}, 
GraphBior filter bank equipped with biorthogonal filter kernels (abbreviated as \textbf{Proposed-Bior}) \citep{narang2013}, 
Generalized graph filter bank equipped with the biorthogonal filter kernels (abbreviated as \textbf{Proposed-GGFB}) \citep{pavez2023}, and Graph filter bank with graph frequency domain sampling (abbreviated as \textbf{Proposed-SS}) \citep{sakiyama2019}.
Since GraphBior and Graph-QMF require a graph bipartition, we apply the Harary decomposition \citep{narang2012, harary1977} to the graph.
All four methods are implemented on the normalized Laplacian of the unweighted line graph with the following parameters. 
Proposed-Bior uses the biorthogonal filter design with $N_{\mathrm{lo}} = 7$ and $N_{\mathrm{hi}} = 9$ over the spectral interval $[0,2]$. 
Proposed-QMF uses the Meyer wavelet kernel with a Chebyshev polynomial approximation of length $24$. 
Note that we use two-channel filter banks.

Uniform quantization is performed on the transformed coefficients, followed by Huffman coding.
The quantization step was set to $0.001$, $0.005$, $0.01$, $0.05$, $0.1$, $0.5$, $1$.
Since there is no competing transform coding for edge weights other than the proposed method, we compare the reconstruction accuracies with the following preprocessing methods.
To make the role of topology preservation explicit, we divide the comparison methods into two groups. 
\begin{description}[style=unboxed]
    \item [\textbf{Direct compression methods:}] \ 
    \begin{description}[style=unboxed,leftmargin=0cm]
        \item[\textbf{Direct-GFB}]: This simply transforms $\Wm$ into the transformed coefficient $\Qm_c$ via a (node-domain) two-channel graph filter bank as follows.
        \begin{equation}
            \Qm_c = \tilde{\Sm}_c^\top \tilde{\Hm}_c \Wm \quad \mathrm{for\ } c = 0,1,
        \end{equation}
        where $\tilde{\Sm}_c^\top$ is a downsampling matrix for the $c$th channel and $\tilde{\Hm}_c$ is the $c$th filter of a graph filter bank.
        We use the GraphBior \citep{narang2013} for the graph Laplacian $\tilde{\Lm}$ of the unweighted graph.
        Transformed coefficients $\Qm_c$ are then quantized and encoded as in the proposed method.
        \item[\textbf{Direct-W}]: This method directly applies uniform quantization and sparse representation to the weighted adjacency matrix $\Wm$ without any transform followed by entropy coding.
    \end{description}

    \item [\textbf{Topology-preserving methods:}] \ 
    \begin{description}[style=unboxed,leftmargin=0cm]
        \item[\textbf{Ordered-DCT}]: The edge weights are sorted in descending order to form a one-dimensional sequence.
        Let $\piv$ denote the permutation satisfying $[\wv]_{\pi_1} \geq [\wv]_{\pi_2} \geq \ldots \geq [\wv]_{\pi_{|\Ec|}}$.
        The reordered edge-weight vector $\wv_\pi$ is then defined as $[\wv_\pi]_\alpha = [\wv]_{\pi_\alpha}$.
        A one-dimensional discrete cosine transform (DCT) is applied to $\wv_\pi$, followed by uniform quantization and entropy coding.
        The permutation $\piv$ required to restore the original edge order is also transmitted and included in the bitrate.
        At the decoder, the inverse DCT is first applied, after which the inverse permutation is used to assign the reconstructed weights to their original edges.
        \item[\textbf{Binary}]: As a baseline, we use the binary adjacency matrix $\Am$.
    \end{description}
\end{description}

Reconstruction accuracy is measured using the signal-to-noise ratio (SNR).
Compression efficiency is measured in bits per edge (BPE), defined as the total transmitted bits of the weighted adjacency matrix divided by the number of edges $|\Ec|$.
In this paper, the total bitrate of weighted adjacency matrices is calculated as follows:
Let $B_{\Am}$ be the number of bits required to transmit the binary topology losslessly. 
For the undirected graphs considered here, the upper triangular part of $\Am$ is sparsely represented, and the positions of the ones are Huffman coded together with the corresponding dictionary and metadata for the graph size. 
For a topology-preserving method, let $B_{\mathrm{weight}}$ be the total number of bits required for all encoded edge-weight or transform-coefficient sequences. 
Then, BPE is defined as follows:
\begin{equation}\label{eqn: total_rate}
    B_{\mathrm{total}} = \frac{B_{\Am} + B_{\mathrm{weight}}}{|\Ec|}
\end{equation}
Ordered-DCT additionally includes the bits $B_{\pi}$ required to transmit the actual permutation obtained by sorting the edge weights in descending order.
Direct-GFB and Direct-W do not separately represent the binary topology. 
Their reported rates are therefore the total bits required to encode their transformed matrices or weighted adjacency matrices, including nonzero positions for sparse representation, quantization step, Huffman dictionaries, and metadata, divided by $|\Ec|$; $B_{\Am}$ is not added separately.

\begin{figure}[tp]
\centering
\subfloat[Sensor graph]{\includegraphics[width=0.5\linewidth]{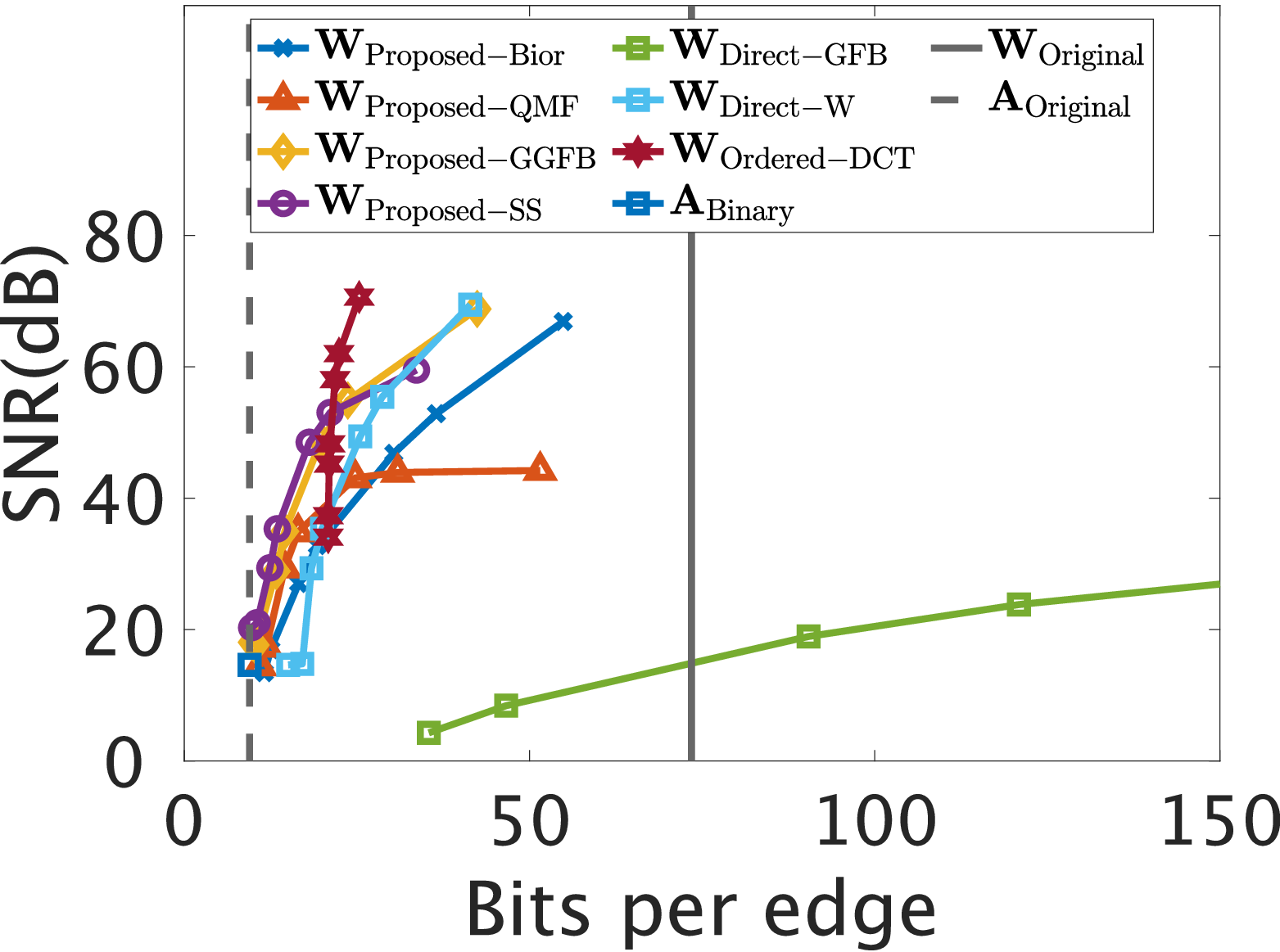}}
\subfloat[$10$-NN graph]{\includegraphics[width=0.5\linewidth]{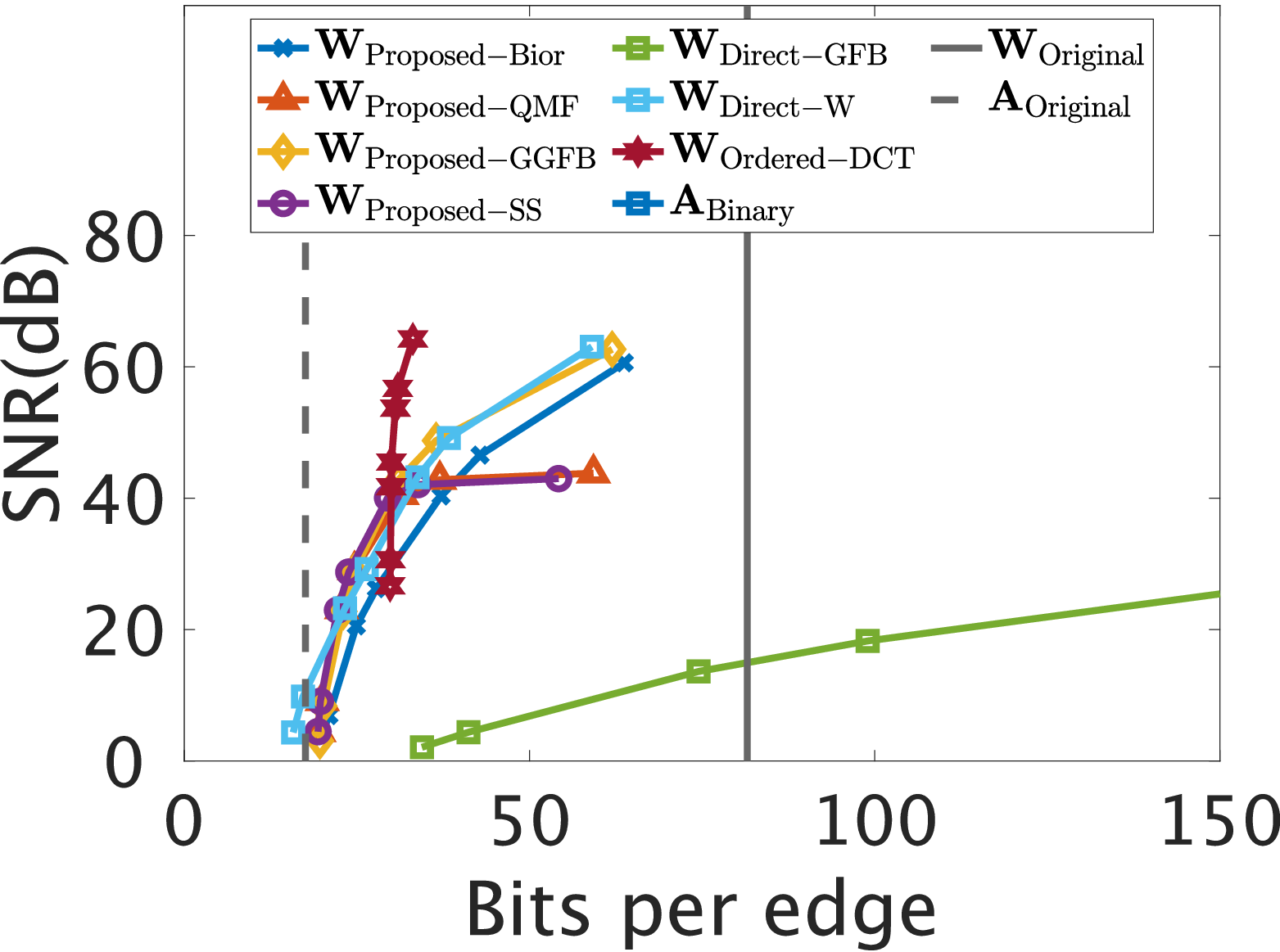}}
\caption{Performance comparison of SNRs of the reconstruction error in dB. Averaged results after $10$ runs are shown. The horizontal axis shows the total transmitted bits per edge. The black solid and dashed lines indicate the bits per edge when the weighted and binary adjacency matrices are losslessly compressed, respectively.}
\label{fig: exp1}
\end{figure}

\subsubsection{Results}
Fig.~\ref{fig: exp1} compares the reconstruction errors for the random sensor graph and 10-NN graph, where the edge weights are determined by the distance-based function.
The proposed methods yield higher SNR and smaller BPEs than alternative preprocessing methods across all compression ratios.
The alternative approaches result in larger BPEs than simply storing the weighted adjacency matrix $\Wm$ after lossless compression ($\Wm_{\mathrm{Original}}$).
This result demonstrates that when edge weights are determined by physical properties such as geometric distance, the proposed framework effectively exploits the inherent smoothness of edge weights on the line graph, resulting in superior reconstruction performance.
Direct-W, which directly quantizes the weighted adjacency matrix without any transform, achieves lower SNRs than the proposed methods because it cannot exploit the inter-edge relationships captured by the line graph structure.
Since the Ordered-DCT is a method that performs the 1-D DCT after reordering, the remaining difference between this baseline and the proposed method further demonstrates the benefits of utilizing the relationships among edges.
Direct-GFB also shows limited performance since the node-domain graph filter bank is not designed to compress edge weights.

\begin{figure}[tp]
\centering
\subfloat[ER model ($\sigma = 0.01$)]{\includegraphics[width=0.5\linewidth]{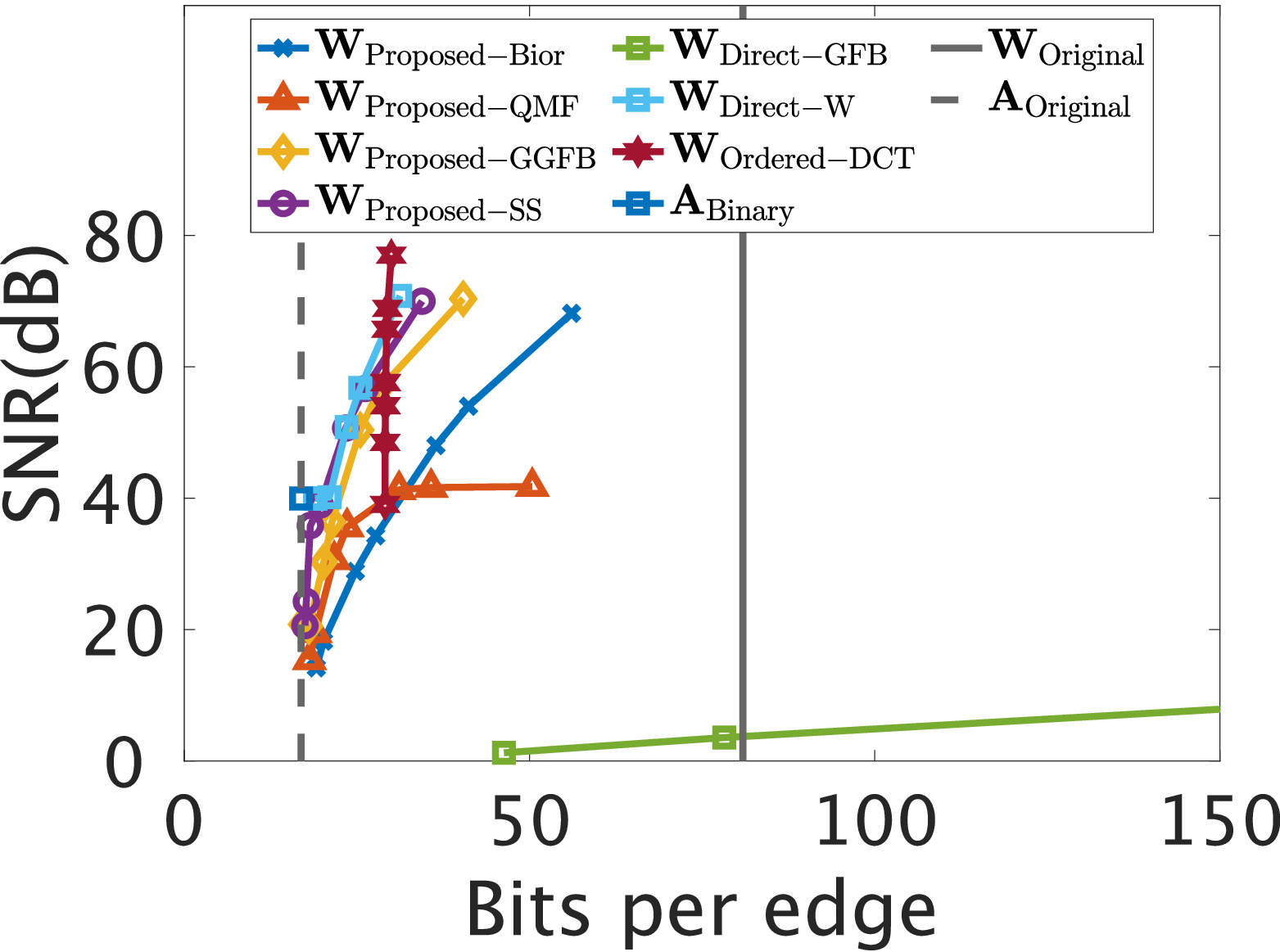}}
\subfloat[ER model ($\sigma = 0.05$)]{\includegraphics[width=0.5\linewidth]{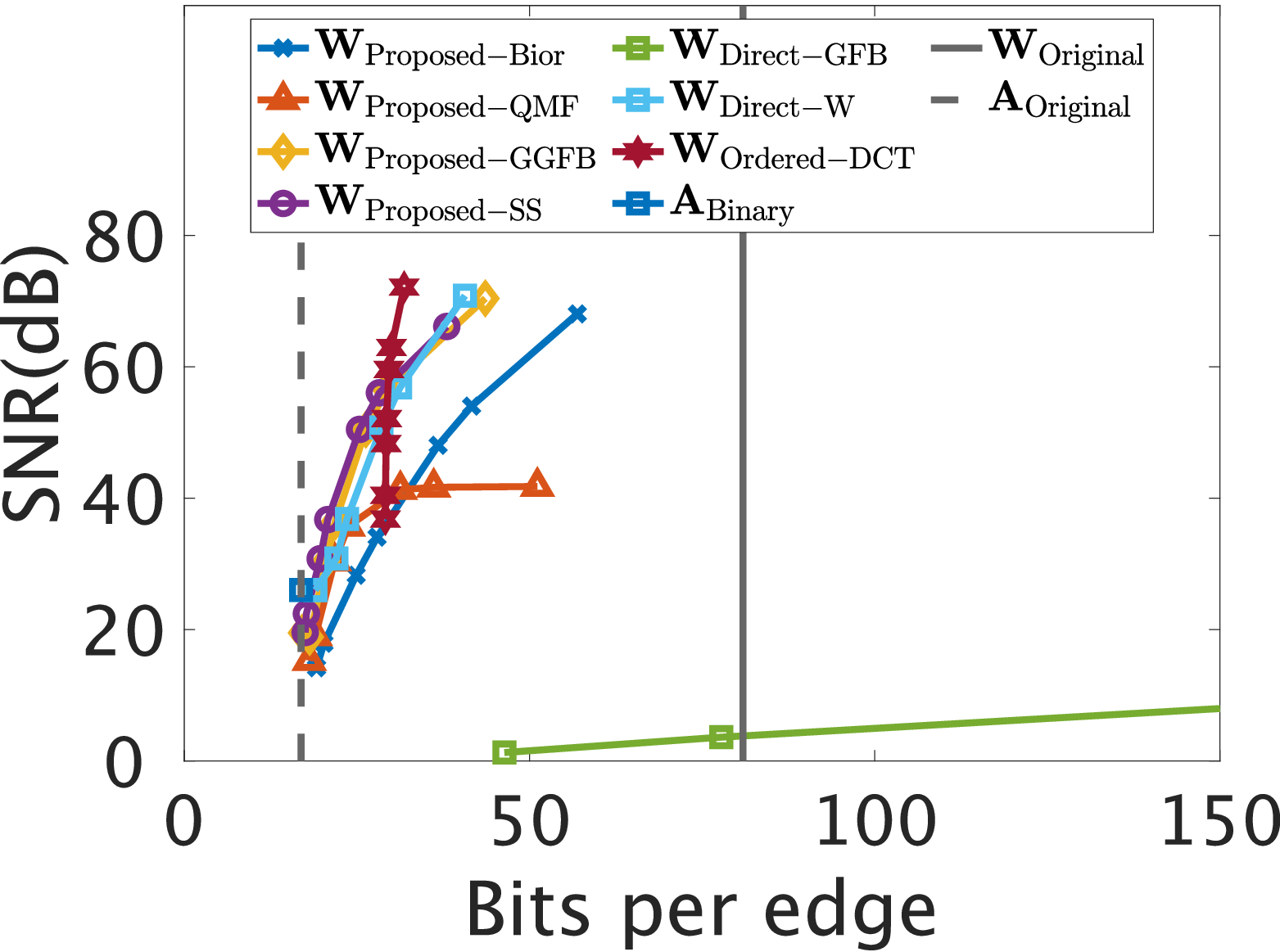}}\\
\subfloat[ER model ($\sigma = 0.1$)]{\includegraphics[width=0.5\linewidth]{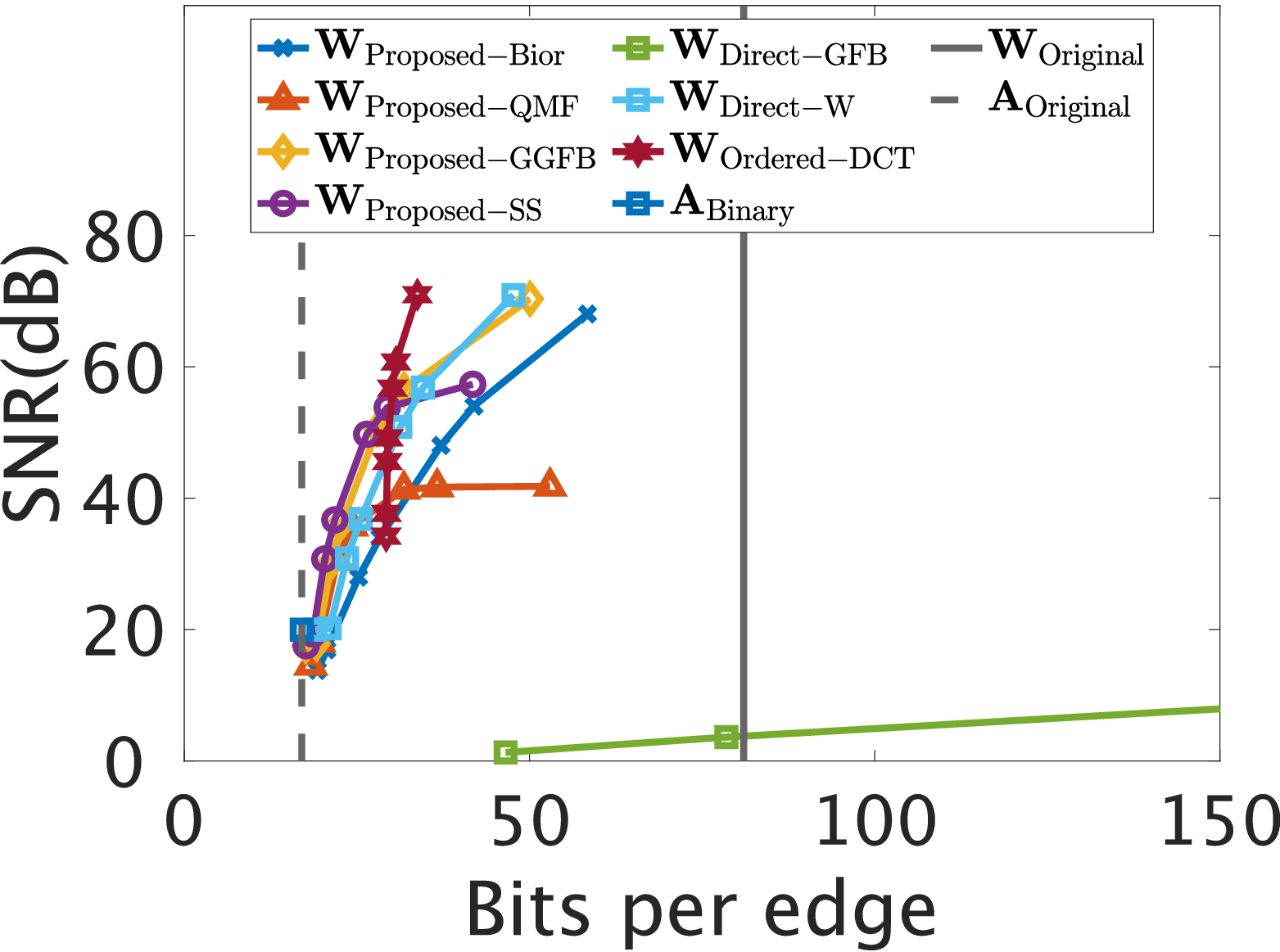}}
\subfloat[ER model ($\sigma = 0.5$)]{\includegraphics[width=0.5\linewidth]{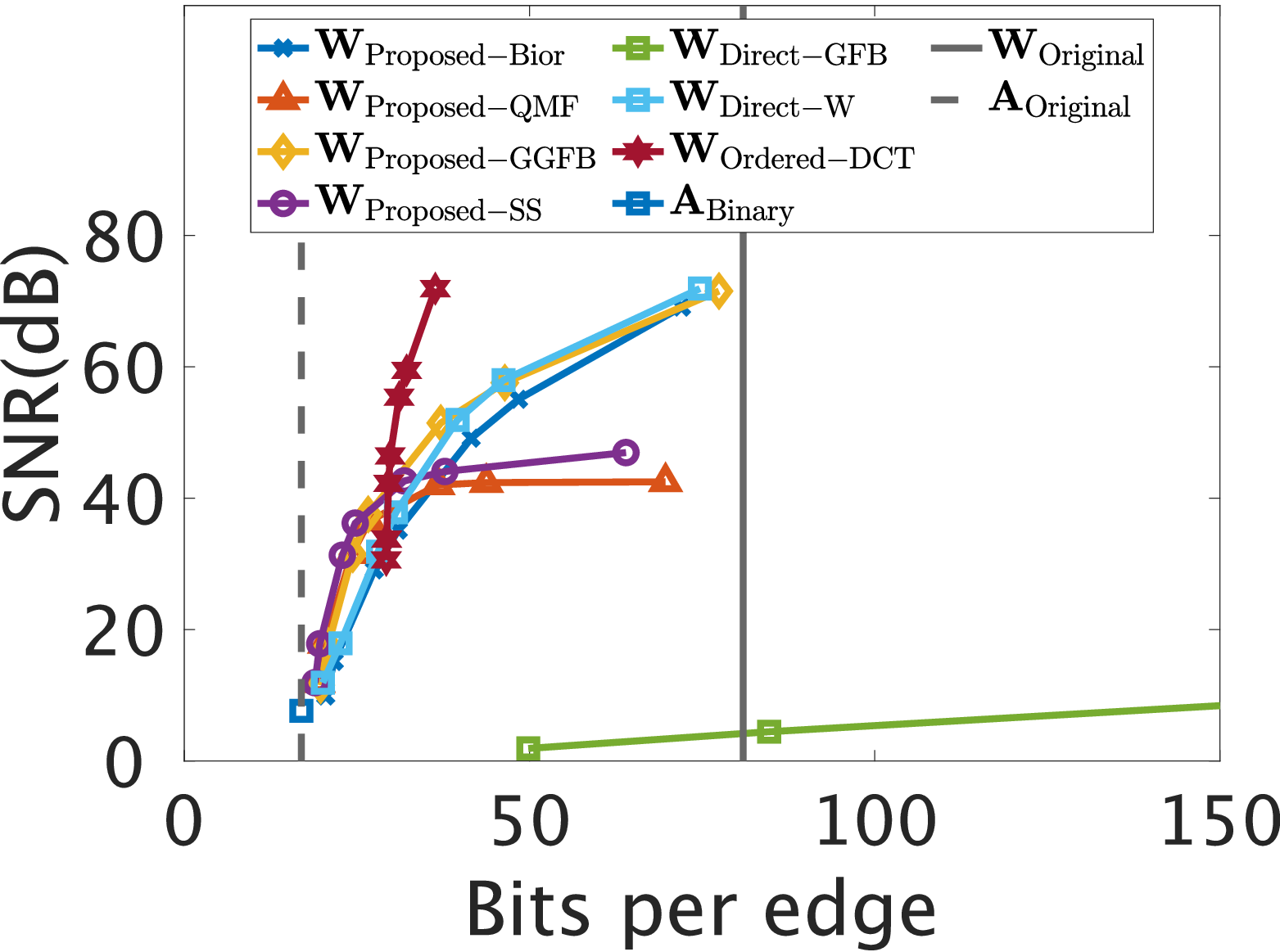}}
\caption{Performance comparison of SNRs of the reconstruction error in dB. Averaged results after $10$ runs are shown. The horizontal axis shows the total transmitted bits per edge. The black solid and dashed lines indicate the bits per edge when the weighted and binary adjacency matrices are losslessly compressed, respectively.}
\label{fig: exp_ER}
\end{figure}

The performance among the proposed methods is similar but slightly different due to the type of graph filter banks.
Proposed-SS achieves better reconstruction accuracy than the other filter banks because its graph-frequency-domain sampling explicitly handles aliasing.
Since Proposed-GGFB does not require graph bipartitioning, its performance is better than that of Proposed-Bior and Proposed-QMF.

Fig.~\ref{fig: exp_ER} shows the results for the ER model with varying edge smoothness controlled by the variance parameter $\sigma$.
As $\sigma$ increases, the edge weights become less smooth on the line graph, making compression more challenging.
Even in such cases, the proposed methods consistently outperform the alternative methods across all tested smoothness levels.
This indicates that the proposed framework achieves high reconstruction accuracy regardless of the compressibility of the edge weights.
When $\sigma$ is small (e.g., $\sigma = 0.01$), the edge weights are nearly constant, resulting in high SNRs for all methods.
However, as $\sigma$ increases and the edge weights become less smooth, the performance gap between the proposed methods and the alternatives widens.
This behavior validates our smoothness assumption discussed in Section \ref{sec: edge_smoothness}: less-smooth edge weights are more difficult to compress, yet the proposed line graph-based transform coding approach still provides significant advantages over direct compression methods.
Furthermore, it is observed that as $\sigma$ increases, the BPE required to achieve a comparable SNR also increases even for the same topology.
This observation confirms that the edge smoothness $\Delta_{\Lm_e}(\wv)$ serves as an effective measure of compression difficulty: graphs with larger edge weight variations require more bits per edge to achieve the same reconstruction quality.

\begin{figure}[tp]
\centering
\subfloat[Australia-Melbourne]{\includegraphics[width=0.5\linewidth]{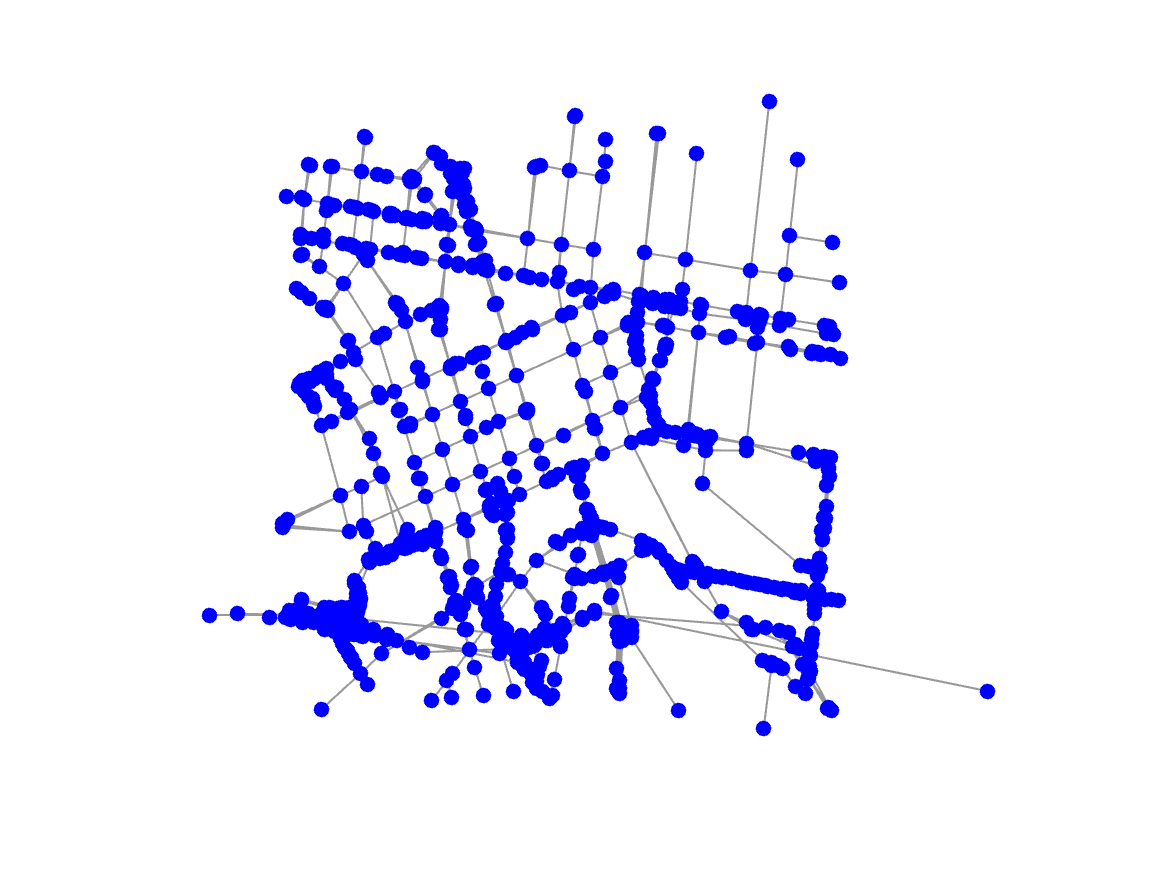}}
\subfloat[Brazil-Sao Paulo]{\includegraphics[width=0.5\linewidth]{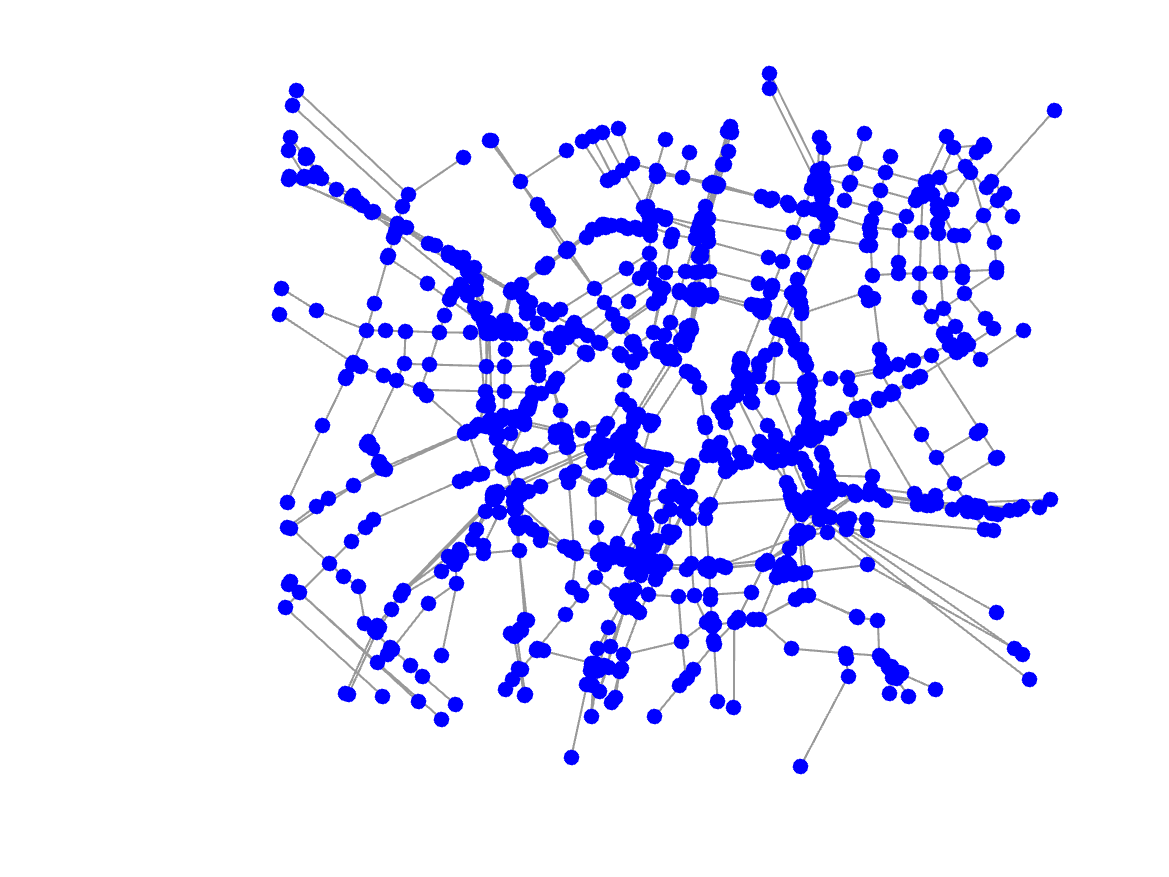}}\\
\subfloat[China-Shanghai]{\includegraphics[width=0.5\linewidth]{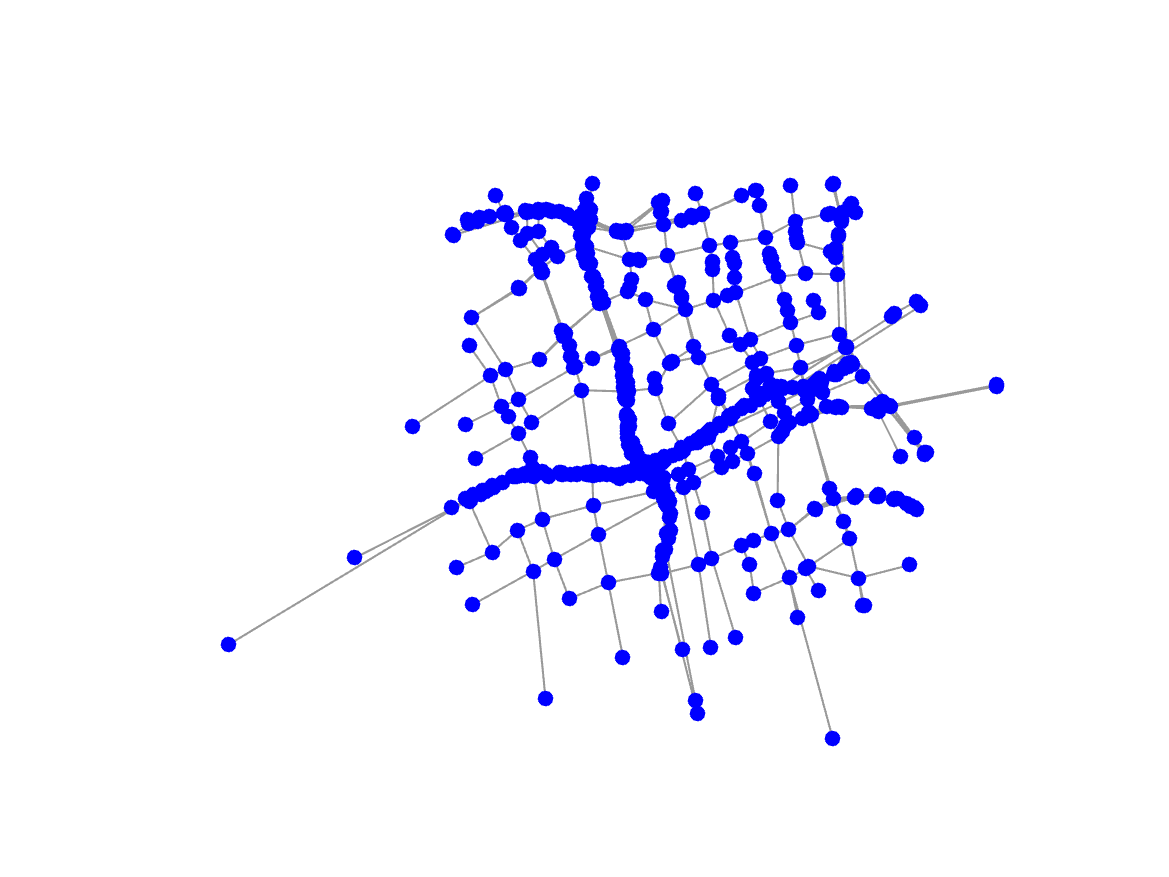}}
\subfloat[Germany-Cologne]{\includegraphics[width=0.5\linewidth]{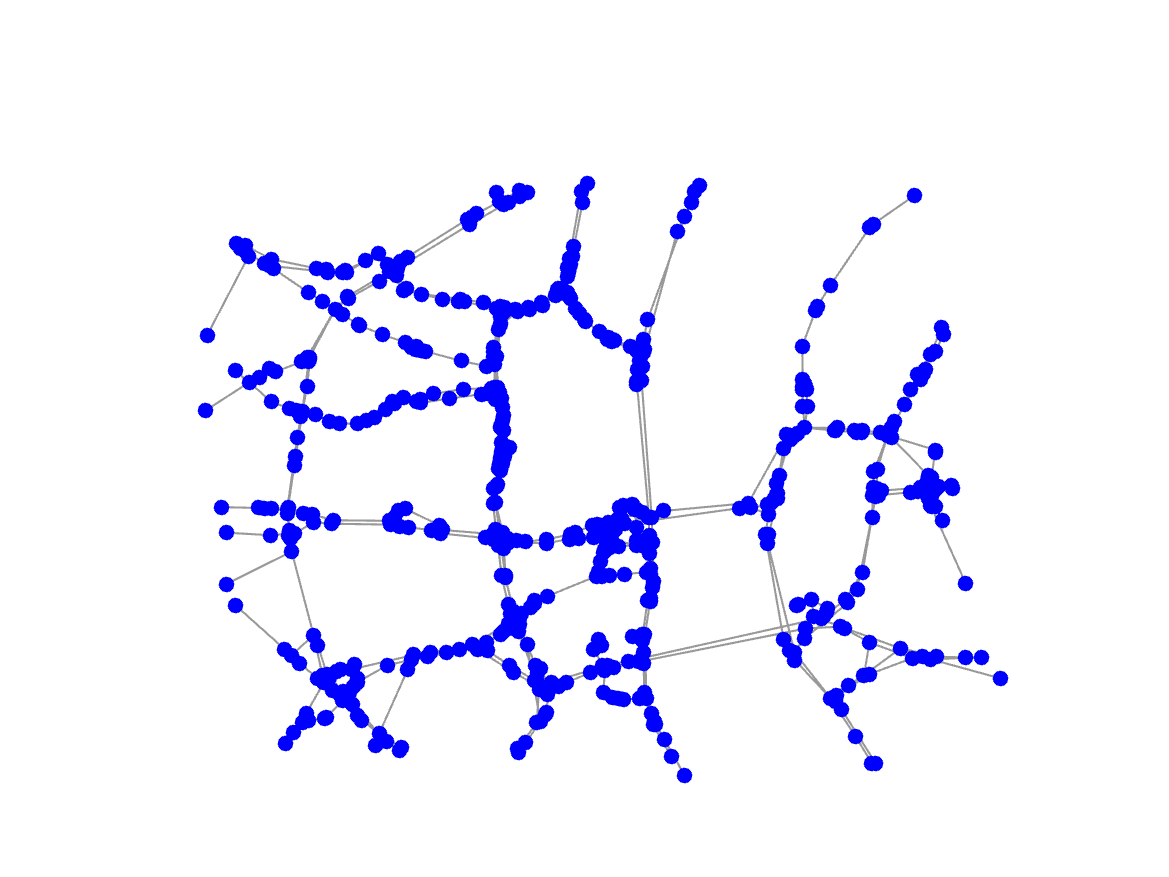}}\\
\subfloat[UAE-Dubai]{\includegraphics[width=0.5\linewidth]{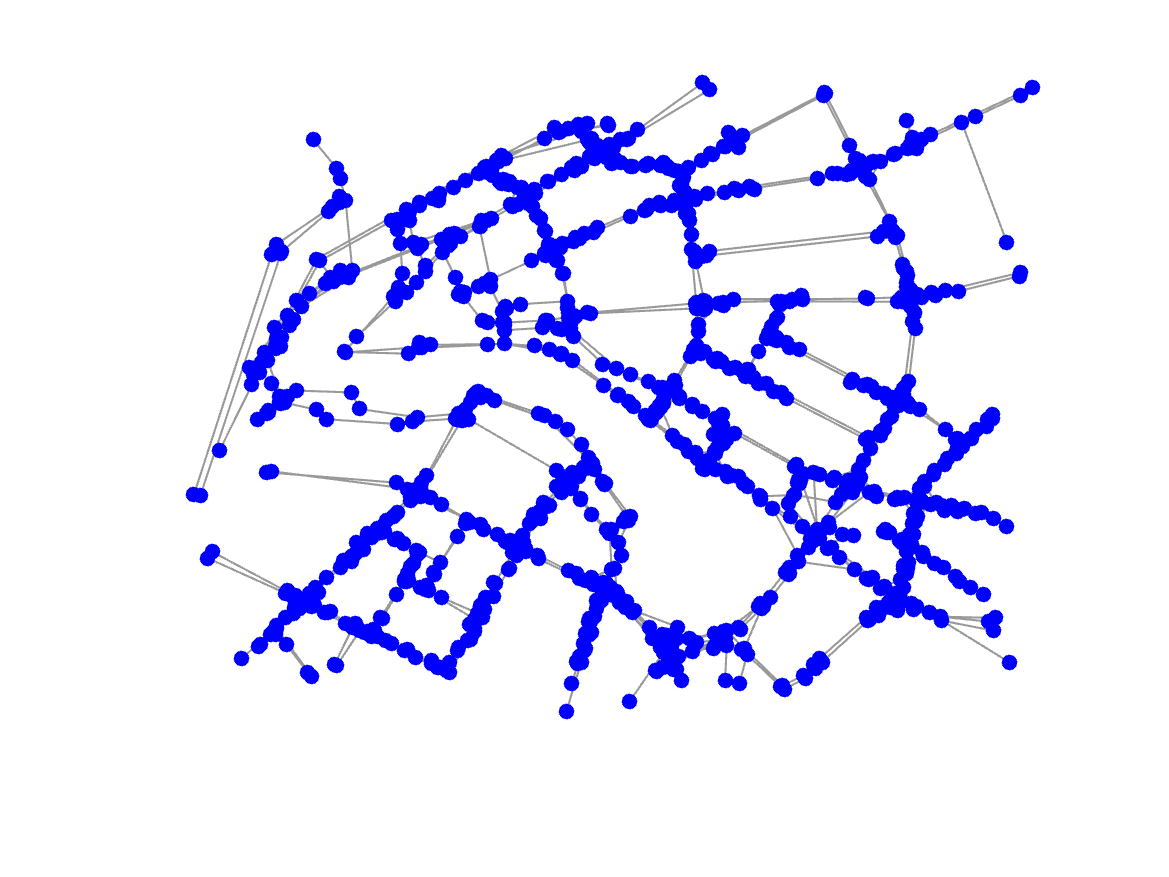}}
\subfloat[USA-Chicago]{\includegraphics[width=0.5\linewidth]{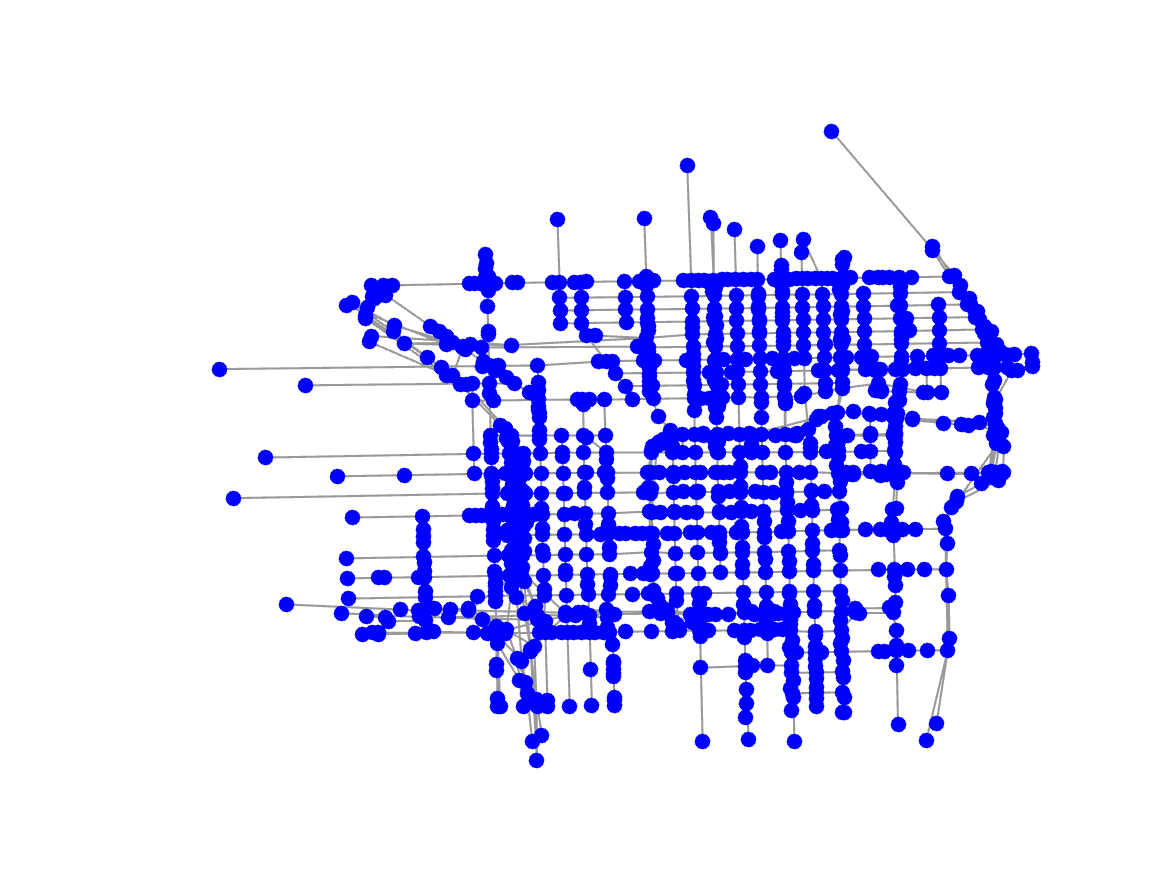}}

\caption{Traffic networks used in the real-world experiment. Each sub-caption denotes a pair of country and city names.}
\label{fig: ex-graph_real}
\end{figure}

\begin{figure}[tp]
    \centering
    \subfloat[Power-grid (WithoutTap)]{\includegraphics[width=0.3\linewidth]{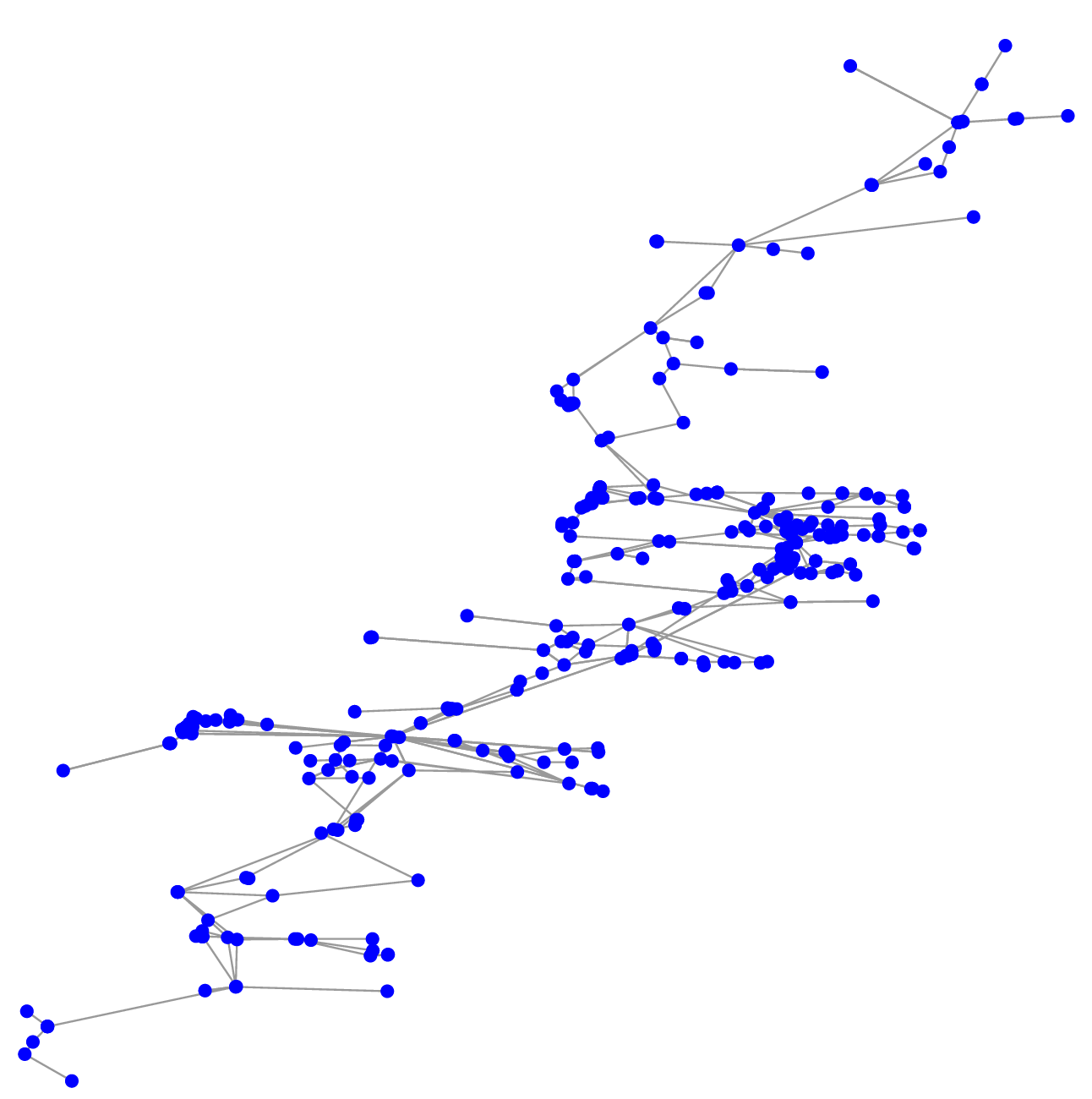}}
    \subfloat[Power-grid (WithTap)]{\includegraphics[width=0.3\linewidth]{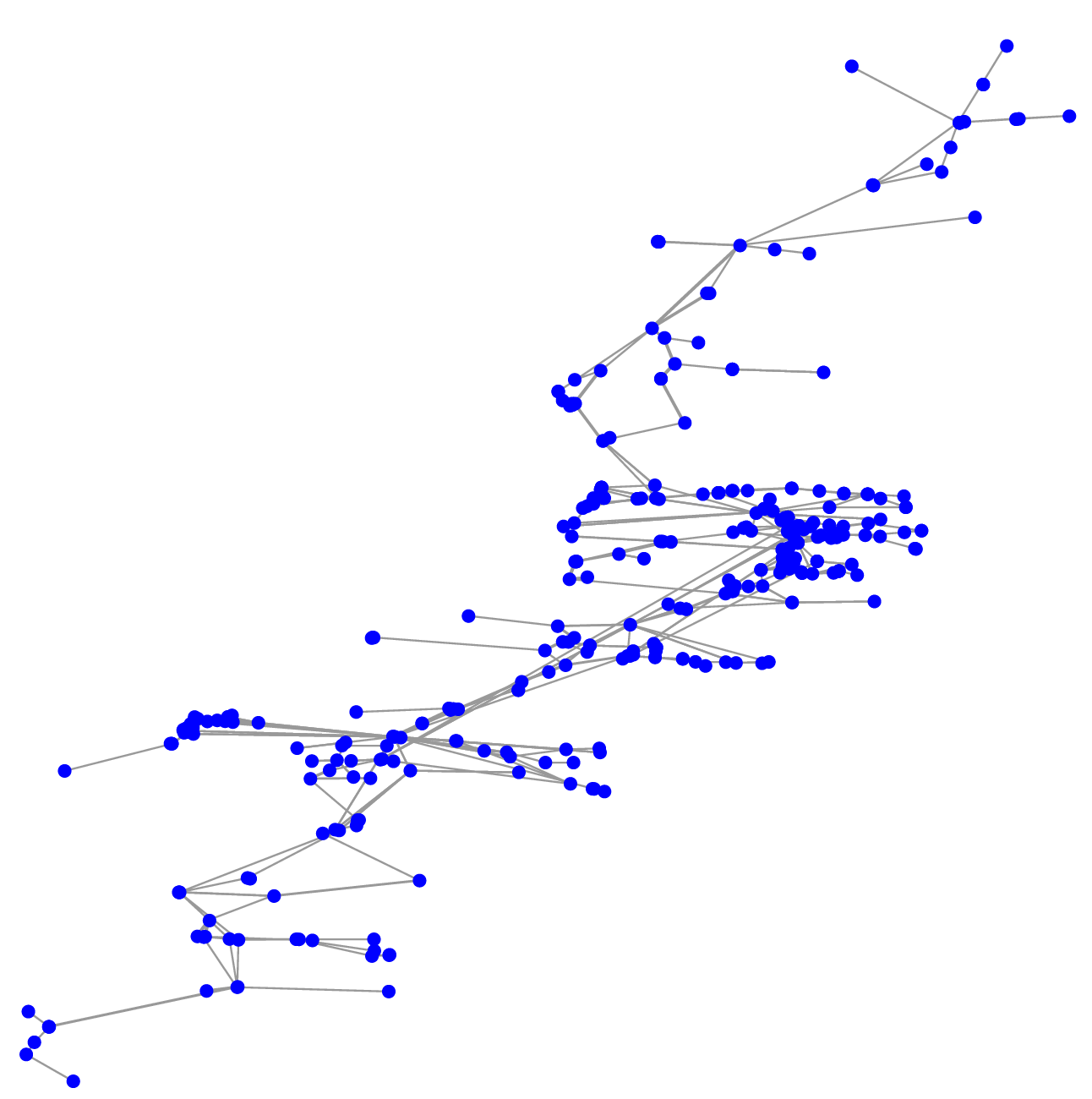}}
    \subfloat[Power-grid (Reduced)]{\includegraphics[width=0.3\linewidth]{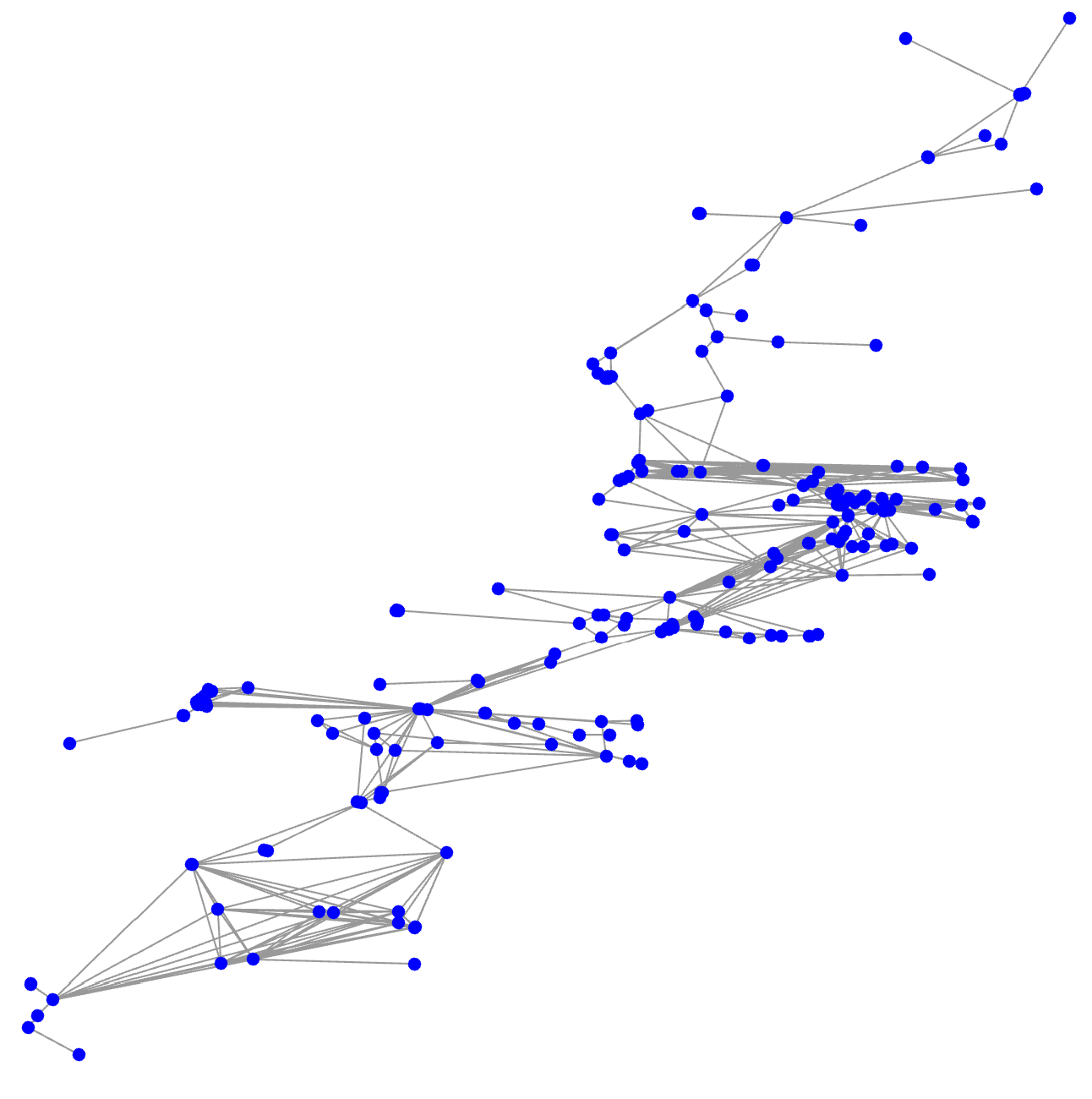}}
    \caption{Chilean power-grids used in the real-world experiment. Each sub-caption denotes a graph type in \citep{kim2018}.}
    \label{fig: ex-graph_powergrid}
\end{figure}

\subsection{Real-World Data}
We now conduct experiments with real-world graphs to validate the effectiveness of the proposed method.

\subsubsection{Setup}
For the real-world evaluation, we use a subset of the GSP-Traffic dataset and the $3$ Chilean power-grids with non-geometric edge weight introduced in Section \ref{sec: toy_example}.
We perform the experiment on the $6$ graphs with the largest numbers of nodes among the 195 connected graphs in the traffic dataset, as these constitute some of the most challenging cases for compression due to their larger numbers of edges and the resulting larger line graphs.
Some of the graphs are shown in Fig.~\ref{fig: ex-graph_real} and Fig.~\ref{fig: ex-graph_powergrid}.
For traffic graphs, the edge weights are determined from geometric distance according to \eqref{eqn: distance}, whereas the edge weights in the power-grid are the transmission-line voltage values described in Section \ref{sec: toy_example}.

The quantization steps were set to the same values as in the synthetic data experiments (see Section \ref{sec: synthetic_experiments}), and we quantify performance in the same way (see Section \ref{sec: synthetic_experiments}).

In addition to evaluating the reconstruction accuracy of edge weights, we also assess the impact of compression on a downstream task in the traffic dataset.
Specifically, we consider a graph signal denoising task that uses the total number of vehicles passing through each intersection over a 500-second interval, as included in the GSP-Traffic dataset.
Each network contains $T = 100$ scenarios, and we denote the graph signal matrix as $\Xm \in \mathbb{R}^{N \times T}$.

For a given graph, the denoising performance is measured by $\mathrm{SNR} = 10 \log_{10} \|\Xm\|_F^2 / \|\Xm - \hat{\Xm}\|_F^2$, where $\hat{\Xm} = (\Id_N + 10\Lm)^{-1}\Xm_n$ is the denoised signal and $\Xm_n$ is the noisy graph signal corrupted by additive Gaussian noise drawn from $\Nc(0,\sigma_n^2)$.
In this experiment, we set the noise standard deviation to $\sigma_n \in \{0.15, 0.3\}$.

Let $\mathrm{SNR}_{\mathrm{orig}}$ and $\mathrm{SNR}_{\mathrm{comp}}$ denote the SNR obtained by using the original weighted adjacency matrix and the reconstructed weighted adjacency matrix after compression, respectively.
We then define the difference in SNR (DSNR) as
\begin{equation}
\mathrm{DSNR} = |\mathrm{SNR}_{\mathrm{orig}} - \mathrm{SNR}_{\mathrm{comp}}|.
\end{equation}
Thus, a smaller DSNR indicates that compression results in less degradation of downstream denoising performance.

\begin{figure}[tp]
\centering
\subfloat[Australia-Melbourne]{\includegraphics[width=0.5\linewidth]{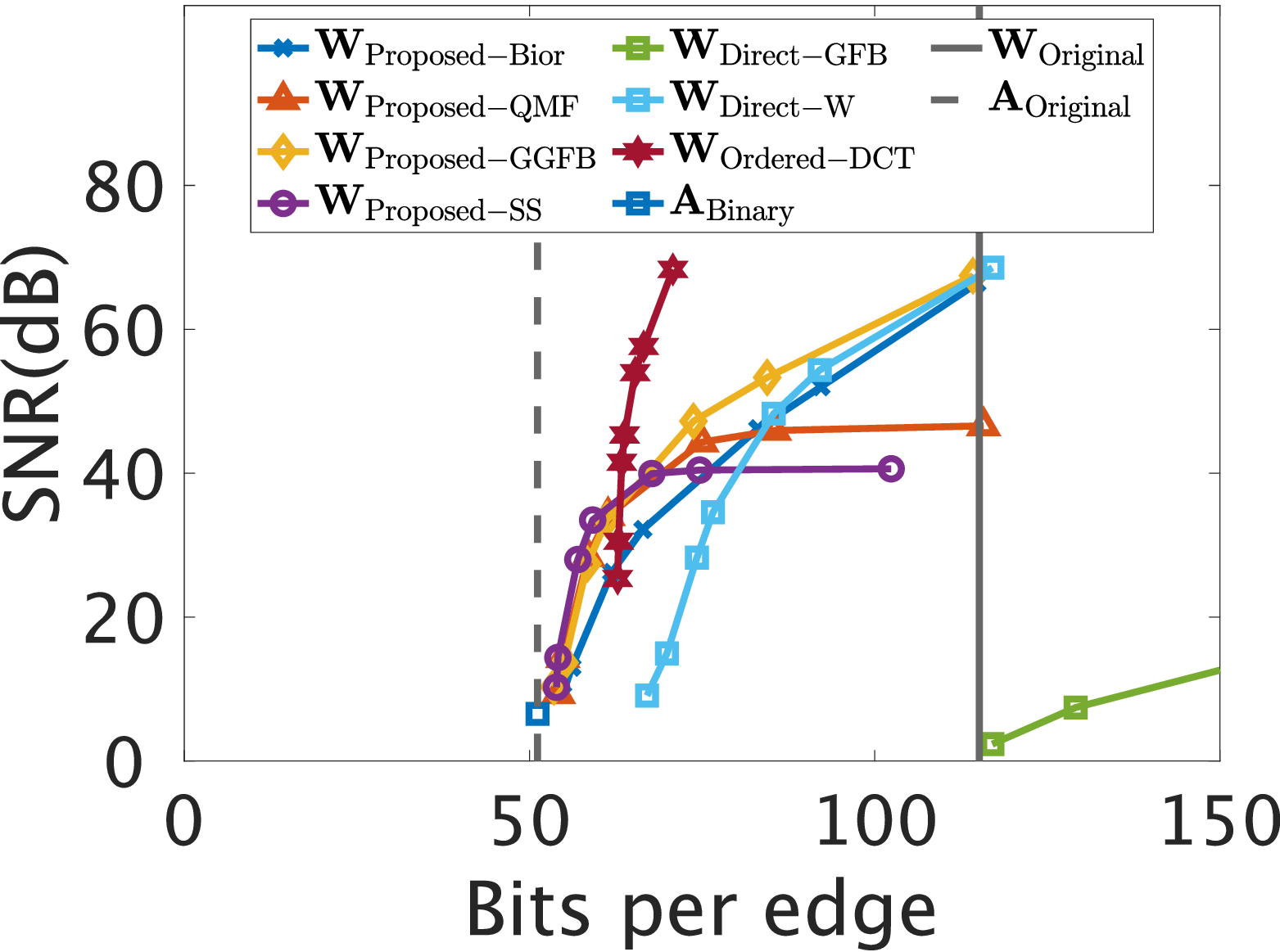}}
\subfloat[Brazil-Sao Paulo]{\includegraphics[width=0.5\linewidth]{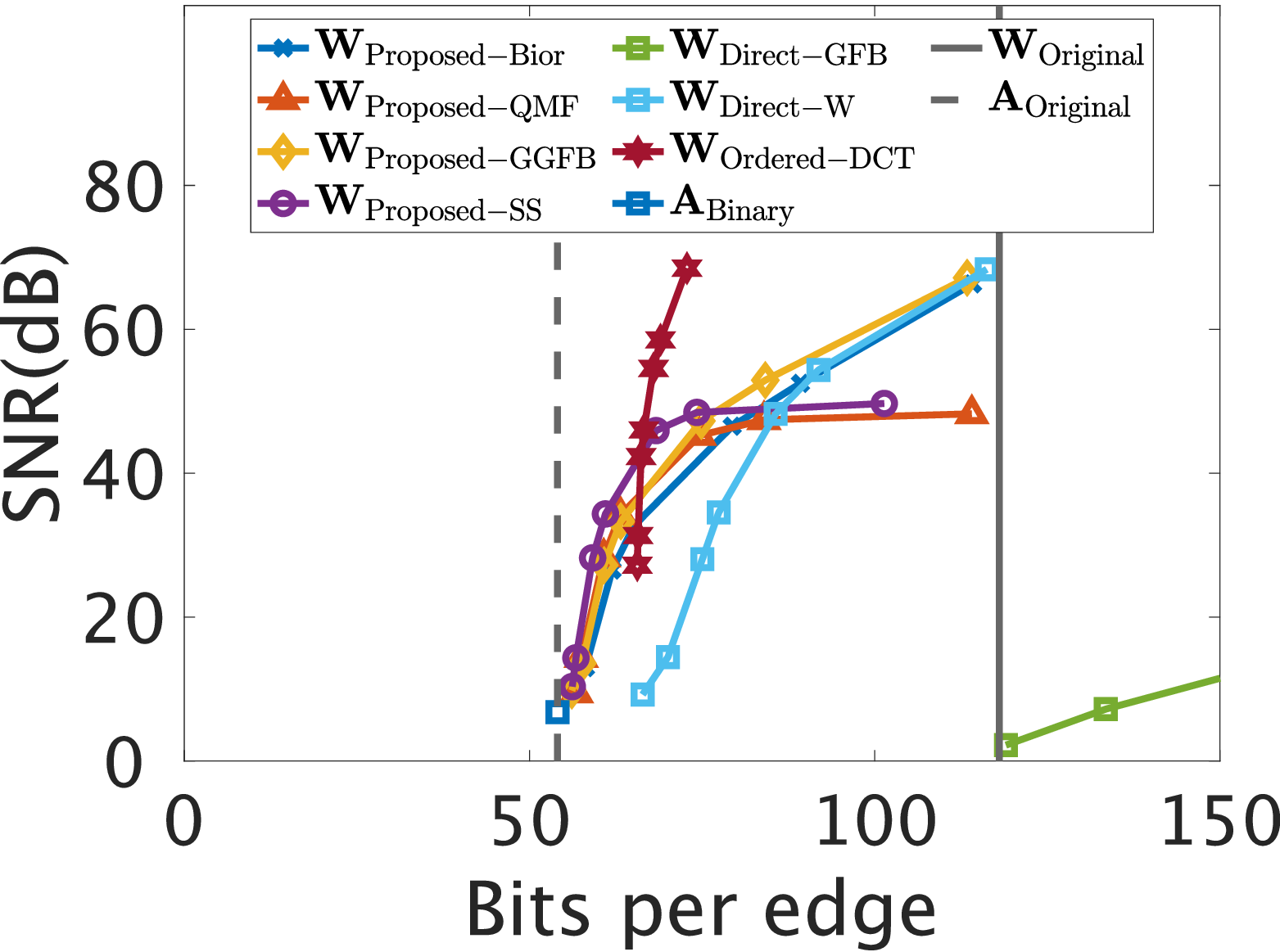}}\\
\subfloat[China-Shanghai]{\includegraphics[width=0.5\linewidth]{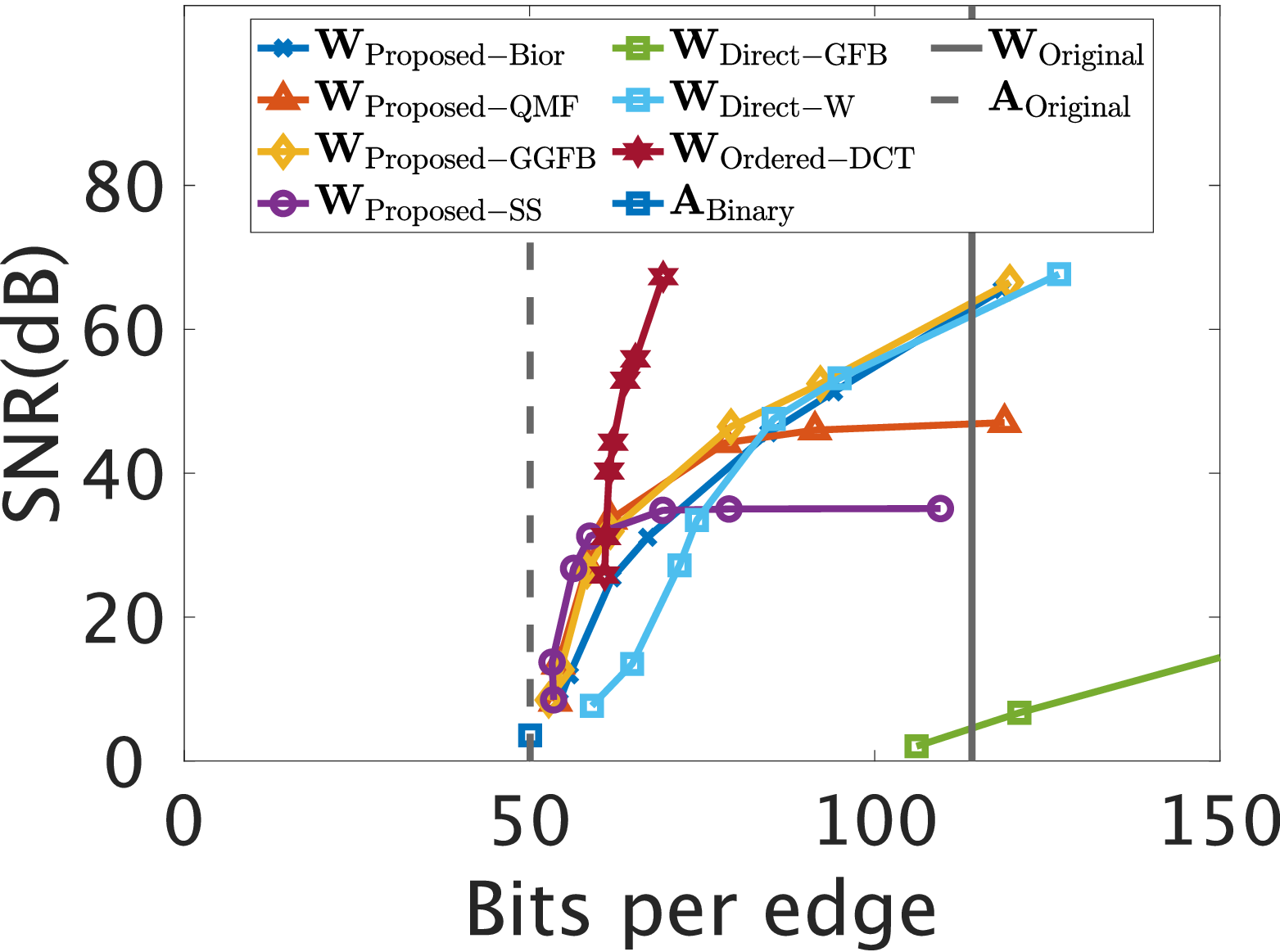}}
\subfloat[Germany-Cologne]{\includegraphics[width=0.5\linewidth]{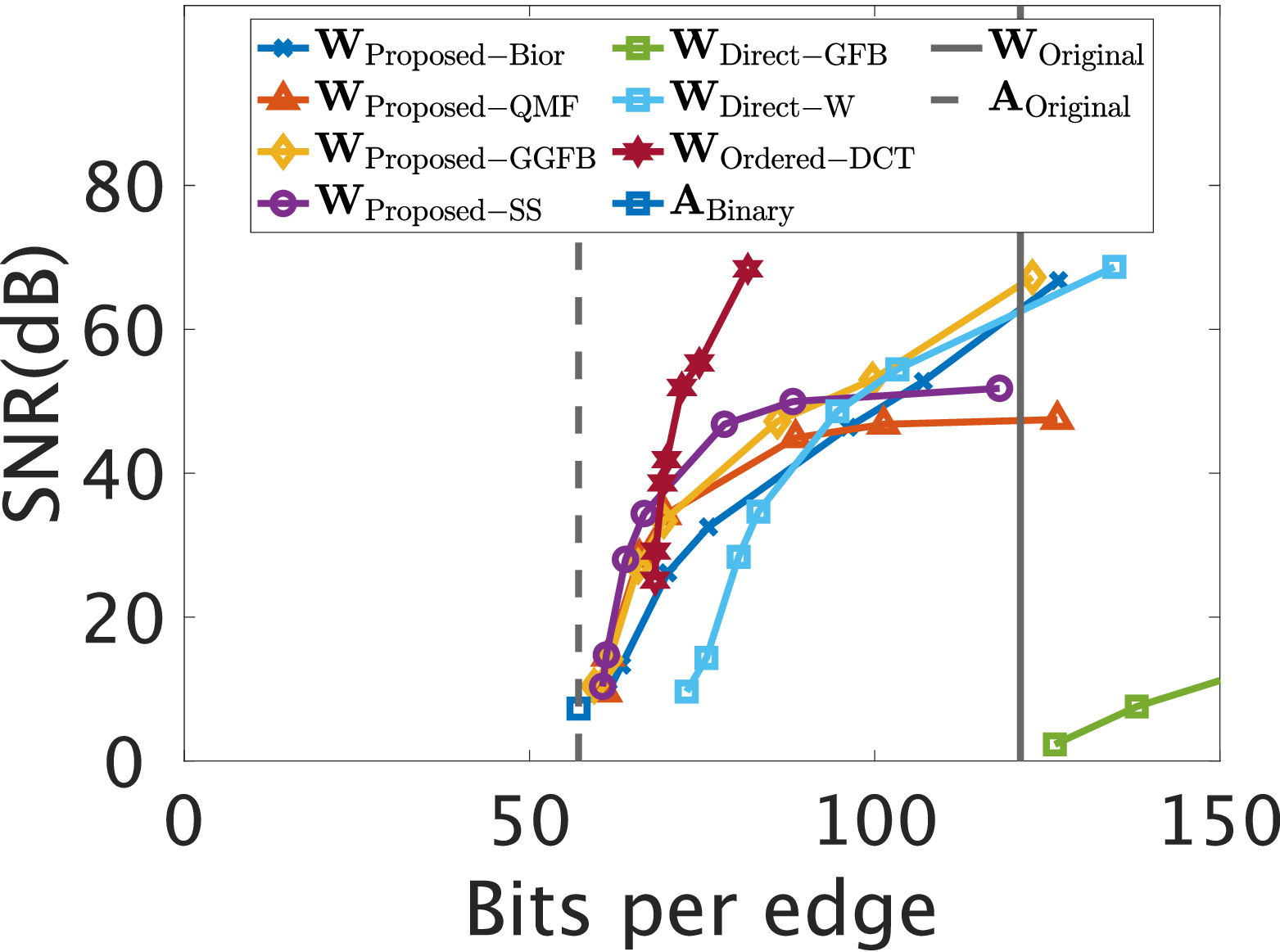}}\\
\subfloat[UAE-Dubai]{\includegraphics[width=0.5\linewidth]{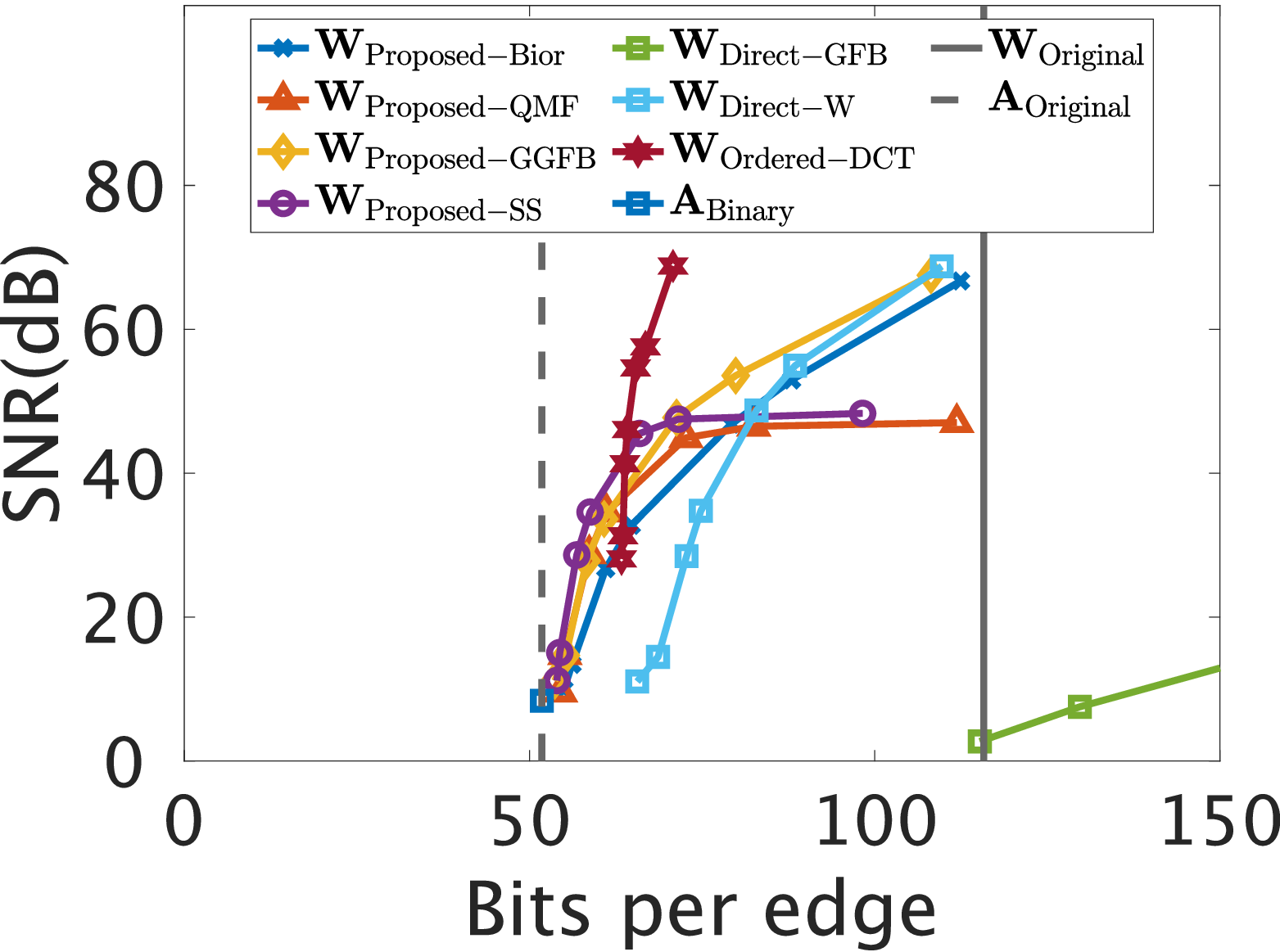}}
\subfloat[USA-Chicago]{\includegraphics[width=0.5\linewidth]{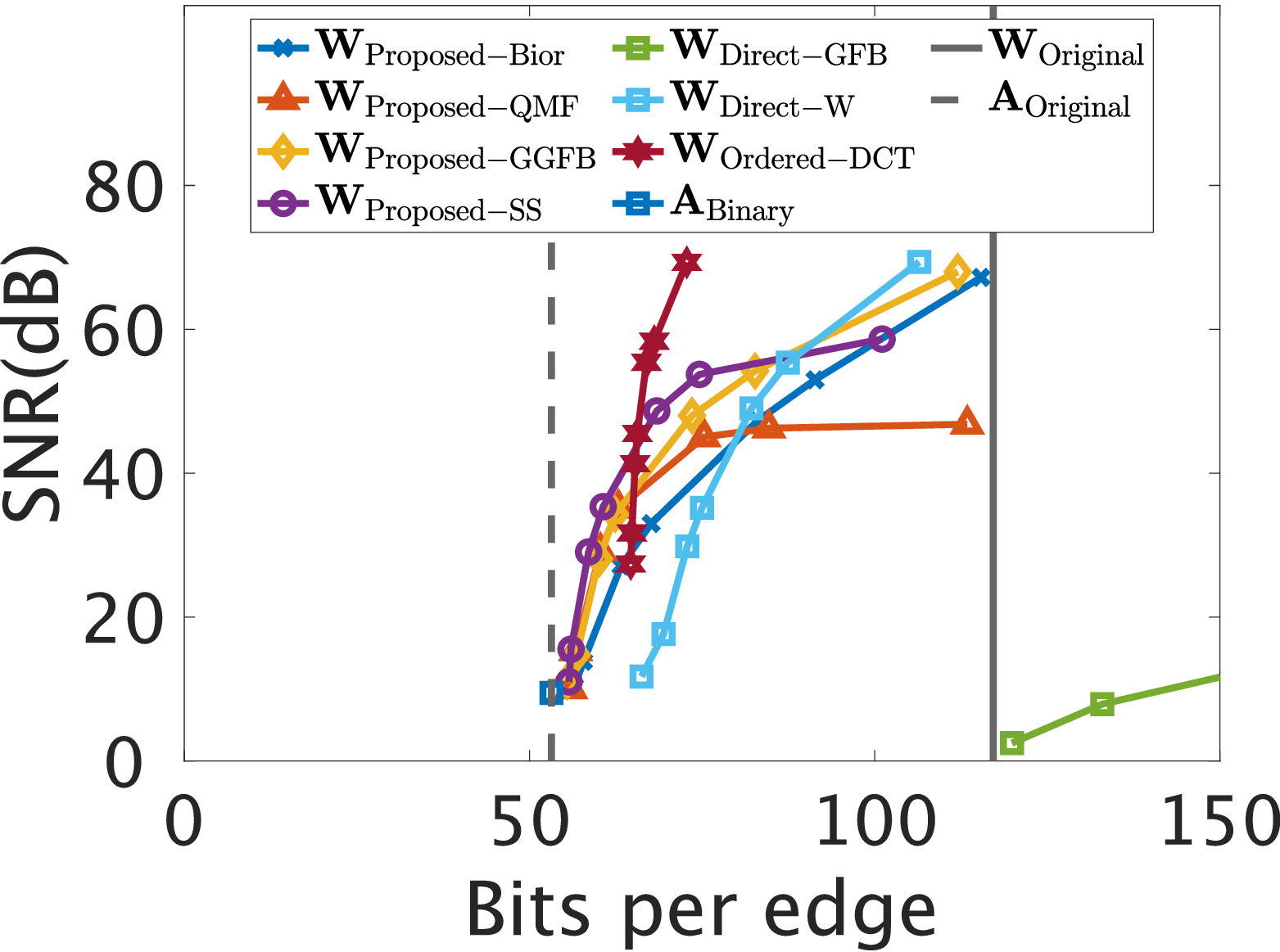}}
\caption{Performance comparison of SNRs of the reconstruction error in dB. The horizontal axis shows the total transmitted bits per edge. The black solid and dotted lines indicate the bits per edge when the weighted and binary adjacency matrices are losslessly compressed, respectively.}
\label{fig: exp_real}
\end{figure}

\begin{figure}[tp]
    \centering
    \subfloat[Power-grid (WithoutTap)]{\includegraphics[width=0.5\linewidth]{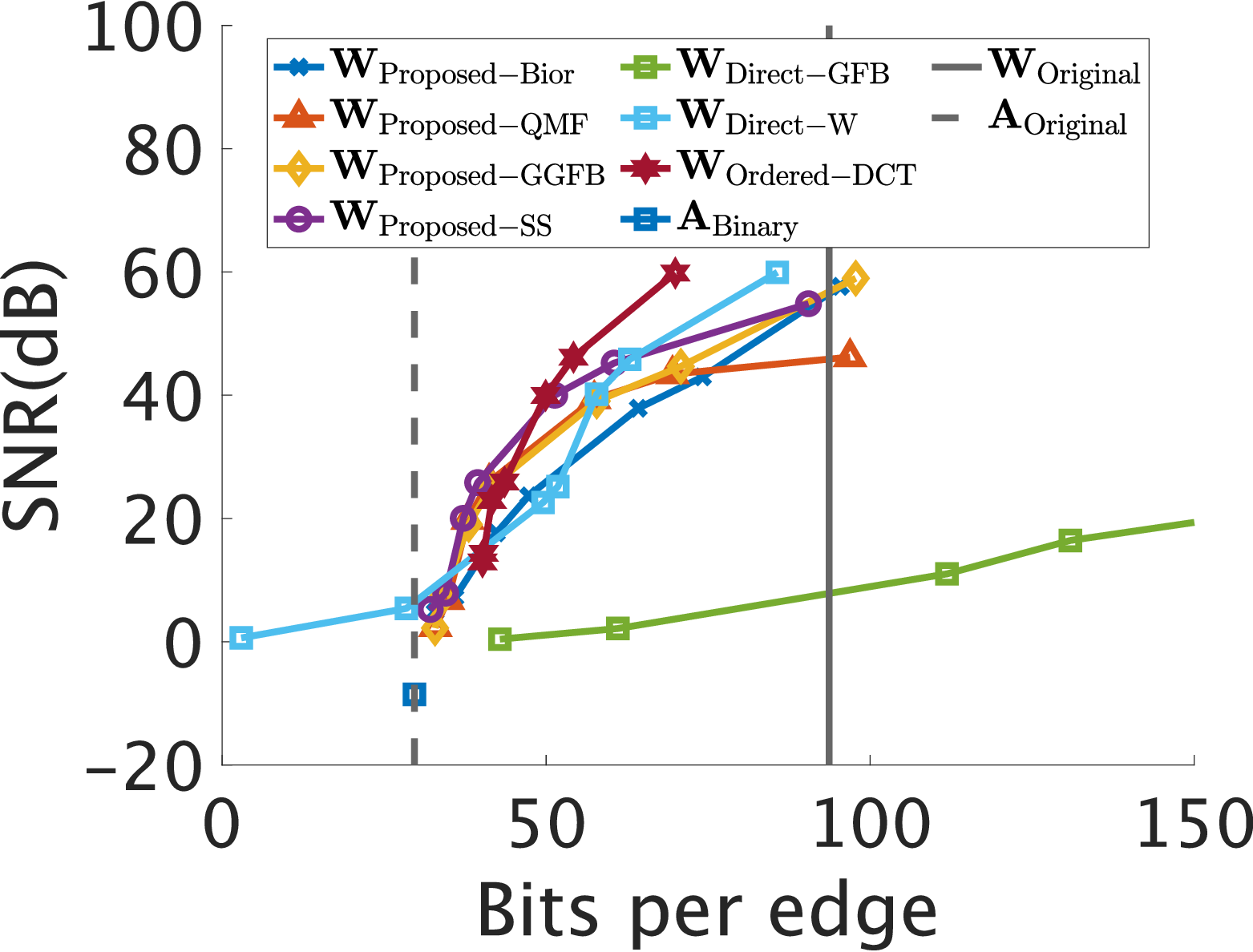}}
    \subfloat[Power-grid (WithTap)]{\includegraphics[width=0.5\linewidth]{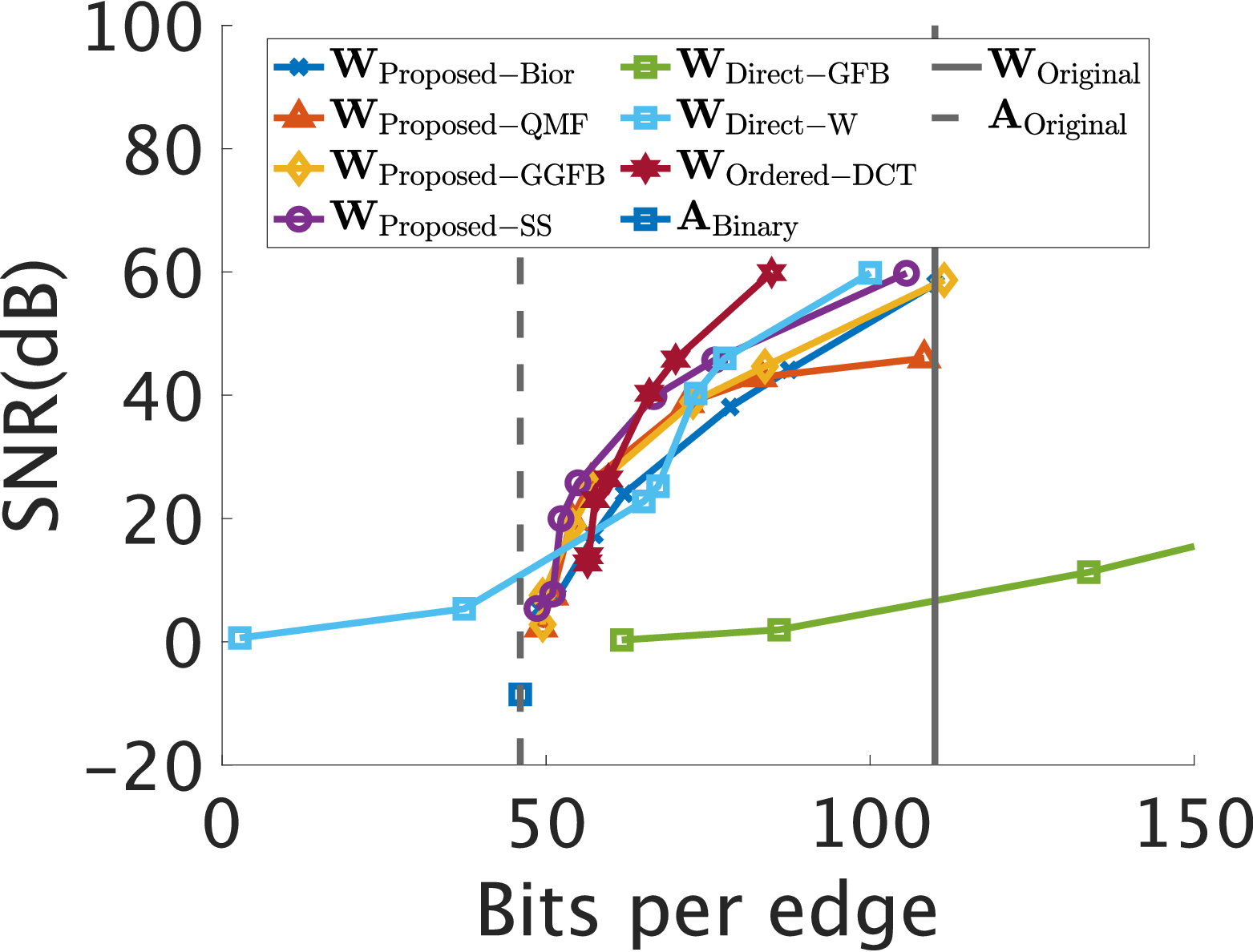}}\\
    \subfloat[Power-grid (Reduced)]{\includegraphics[width=0.5\linewidth]{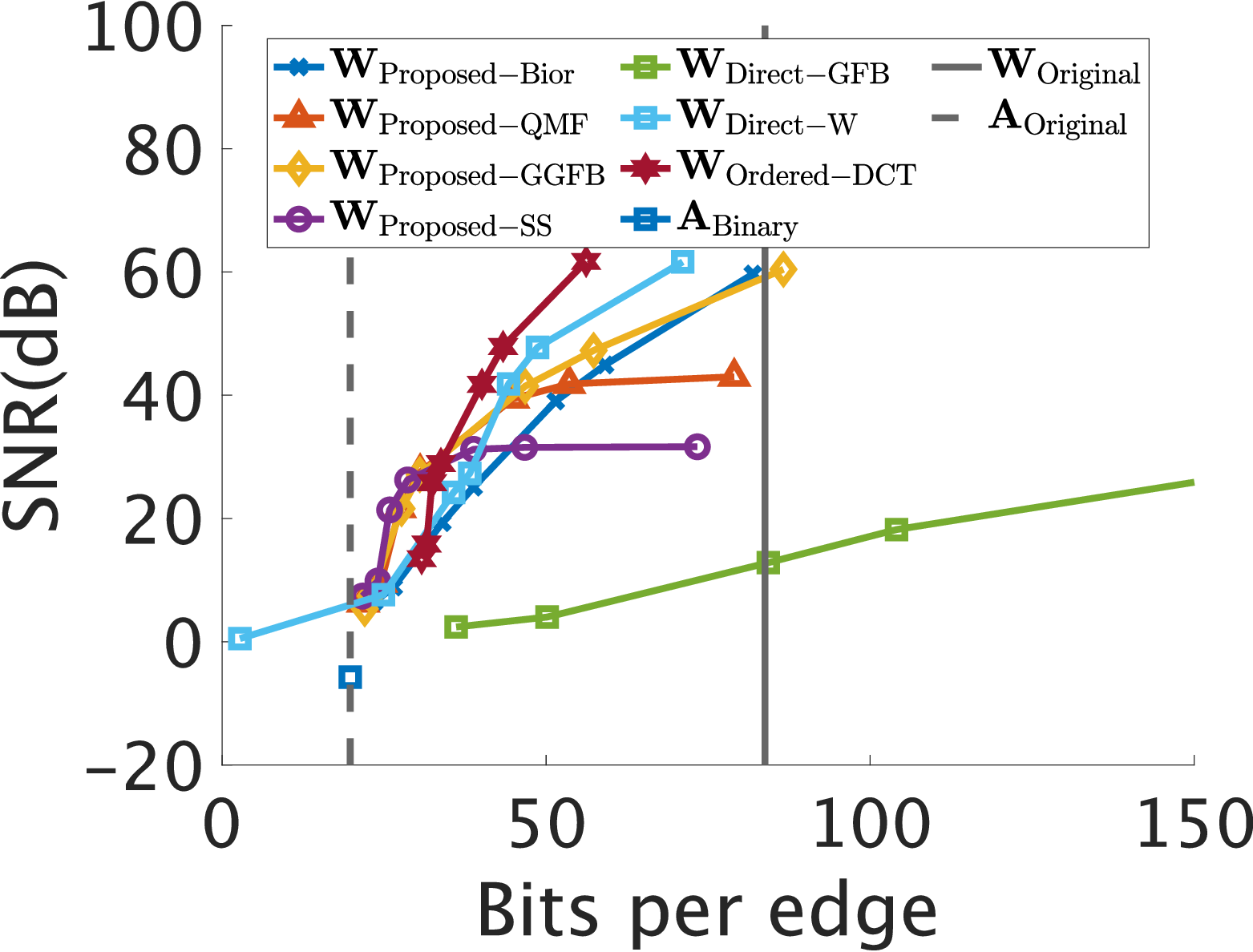}}
    \caption{Performance comparison of SNRs of the reconstruction error in dB. The horizontal axis shows the total transmitted bits per edge. The black solid and dotted lines indicate the bits per edge when the weighted and binary adjacency matrices are losslessly compressed, respectively. The edge weights are transmission-line voltage values rather than geometry-derived distances.}
    \label{fig: exp_powergrid}
\end{figure}

\begin{figure}[tp]
    \centering
    \subfloat[$\sigma_n = 0.15$]{\includegraphics[width=0.5\linewidth]{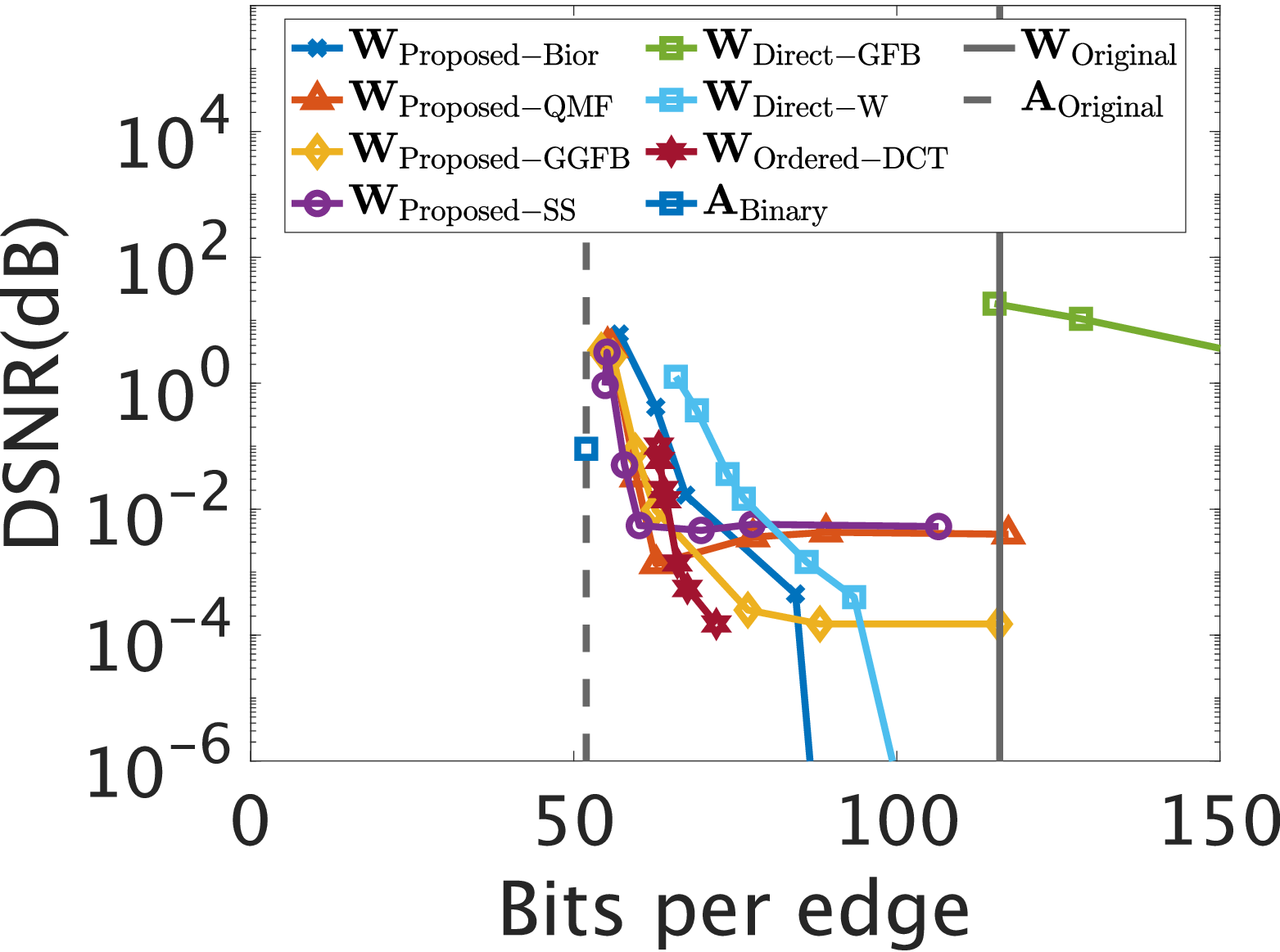}}
    \subfloat[$\sigma_n = 0.3$]{\includegraphics[width=0.5\linewidth]{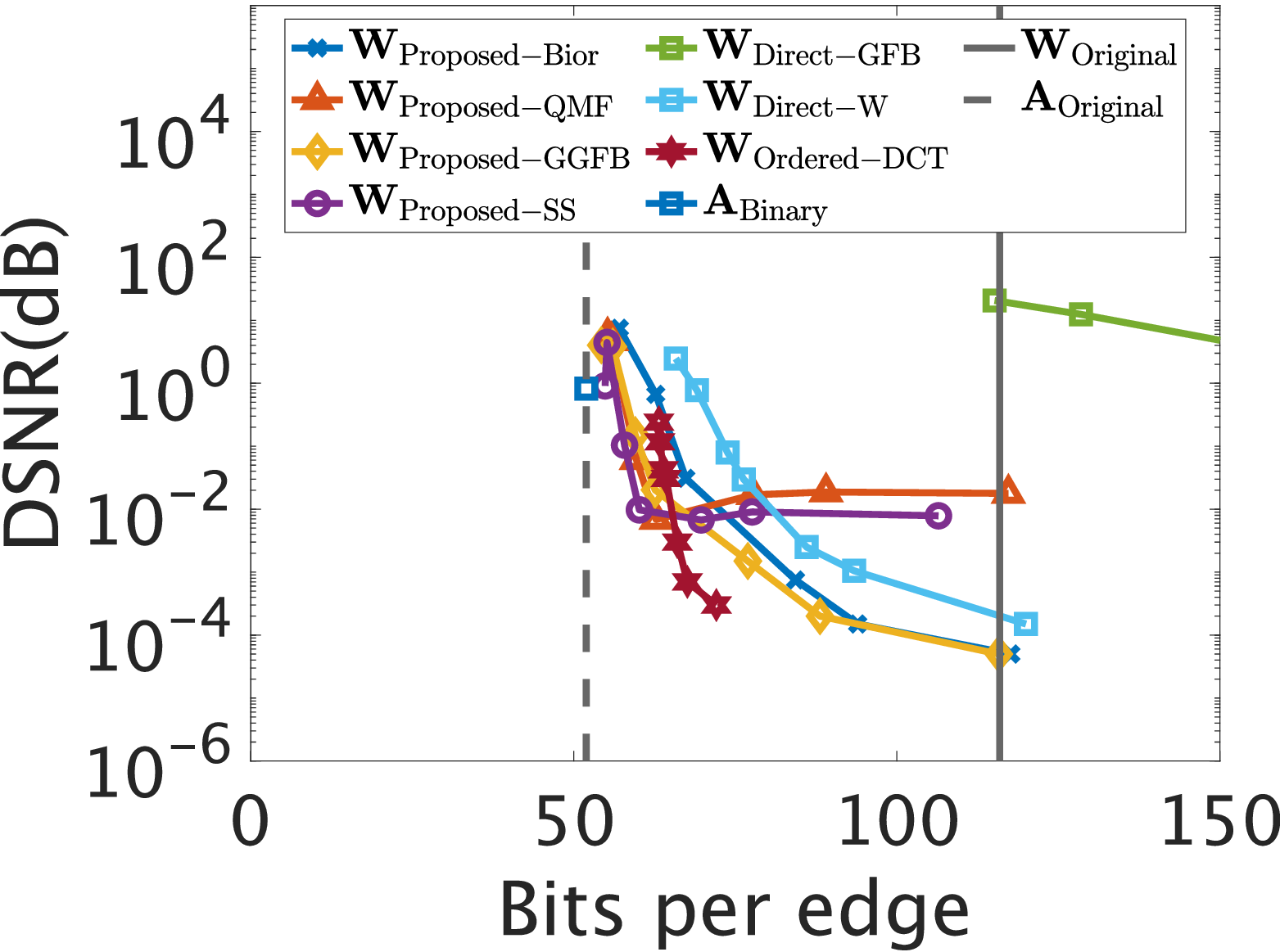}}
    \caption{The performance comparison of DSNRs of the graph signal denoising task. The horizontal axis shows the total transmitted bits per edge. The black solid and dotted lines indicate the bits per edge when the weighted and binary adjacency matrices are losslessly compressed, respectively.}
    \label{fig: exp-real-signal}
\end{figure}

\subsubsection{Experimental Results}
Fig.~\ref{fig: exp_real} shows the results of the traffic graph compression.
Overall, the results are similar to those in the synthetic case: the proposed method achieves higher reconstruction accuracy than the alternative methods across graph structures.

Fig.~\ref{fig: exp_powergrid} shows the compression results for the three power-grid graphs.
Across the three graphs, the proposed methods achieve better reconstruction accuracy than the alternative methods.
These results indicate that the advantage of the proposed framework is not limited to geometrically defined edge weights; the inter-edge relationships in the line graph also benefit non-geometric edge weights.

Fig.~\ref{fig: exp-real-signal} shows the DSNR results for the graph signal denoising task.
When the quantization step is small, all methods maintain denoising performance close to the uncompressed case, yielding DSNR values near zero.
However, as the quantization step increases, some methods exhibit a notable degradation in DSNR.

This degradation occurs because coarse quantization can cause some reconstructed edge weights to become zero.
Even though the binary topology $\Am$ is preserved by lossless compression, coarsely quantizing the weighted adjacency matrix $\Wm$ or the edge-weight vector $\wv$ can collapse nonzero weights to zero.
Consequently, the corresponding edges are removed during reconstruction, resulting in a substantial loss of topological information.

This phenomenon leads to significant performance degradation in the downstream denoising task compared to the uncompressed case.
The proposed methods demonstrate more robust performance against this issue, as graph filter bank-based transform coding on the line graph provides a more compact representation of edge weights, enabling finer quantization at the same bit rate.

\section{Conclusion}\label{sec: conclusion}
This paper proposes a lossy compression framework for weighted adjacency matrices using graph filter banks.
In our method, the binary adjacency information and edge weights of the graph are compressed losslessly and lossily, respectively.
The original graph is transformed into a line graph for lossy compression of edge weights.
Edge weights are treated as a graph signal on the line graph and compressed using graph filter banks.
We also introduce edge smoothness and show that it serves as a measure of the difficulty of compression.
Experimental results demonstrate the effectiveness of the proposed method compared with several transform coding approaches.

\printbibliography
\vfill\pagebreak

\end{document}